%% file: main.tex
\documentclass[sigconf, nonacm]{acmart}
\AtBeginDocument{%
  }

\setcopyright{acmlicensed}
\copyrightyear{2025}
\acmYear{2025}
\setcopyright{cc}
\setcctype{by}
\acmConference[CCS '25]{Proceedings of the 2025 ACM SIGSAC Conference on
Computer and Communications Security}{October 13--17, 2025}{Taipei, Taiwan}
\acmBooktitle{Proceedings of the 2025 ACM SIGSAC Conference on Computer
and Communications Security (CCS '25), October 13--17, 2025, Taipei,
Taiwan}
\acmDOI{10.1145/3719027.3765048}
\acmISBN{979-8-4007-1525-9/2025/10}

\acmSubmissionID{250}

\usepackage{amsmath}
\usepackage{mathtools}
\usepackage{amsthm}
\usepackage{bm}
\usepackage{algorithm}
\usepackage{algorithmic}

\usepackage{thm-restate}

\newtheorem{proposition}{Proposition}

\newtheorem*{lemma*}{Lemma}

\theoremstyle{definition}
\newtheorem{definition}{Definition}

\newtheorem{remark}{Remark}
\newtheorem{example}{Example}

\newcommand{\inprod}[2]{\ensuremath{\langle #1 , \, #2 \rangle}}
\newcommand{\mset}[1]{\ensuremath{ \{\hspace{-2.5pt}\{ #1 \}\hspace{-2.5pt}\} }}

\usepackage[capitalize,noabbrev]{cleveref}
\usepackage{pgffor}
\foreach \x in {a,...,z}{%
\expandafter\xdef\csname vec\x \endcsname{\noexpand\ensuremath{\noexpand\bm{\x}}}
}
\foreach \x in {A,...,Z}{%
\expandafter\xdef\csname vec\x \endcsname{\noexpand\ensuremath{\noexpand\bm{\x}}}
}
\foreach \x in {A,...,Z}{%
\expandafter\xdef\csname c\x \endcsname{\noexpand\ensuremath{\noexpand\mathcal{\x}}}
}
\foreach \x in {A,...,Z}{%
\expandafter\xdef\csname bb\x \endcsname{\noexpand\ensuremath{\noexpand\mathbb{\x}}}
}

\usepackage{xcolor}  
\usepackage{thm-restate}

\begin{document}

\title{Managing Correlations in Data and Privacy Demand}

\author{Syomantak Chaudhuri}
\email{syomantak@berkeley.edu}
\orcid{1234-5678-9012}
\affiliation{%
  \institution{University of California, Berkeley}
  \city{Berkeley}
  \state{CA}
  \country{USA}
}

\author{Thomas A. Courtade}
\email{courtade@berkeley.edu}
\orcid{0000-0001-7106-7358}
\affiliation{%
  \institution{University of California, Berkeley}
  \city{Berkeley}
  \state{CA}
  \country{USA}
}

\input{0-abstract}

\maketitle

\input{1-intro}

\input{2-pd}
\input{3-shdp}
\input{3.5-Privacy}
\input{4-algo}
\input{6-exp}
\input{7-con}

\begin{acks}
S.C. acknowledges  the support of AI Policy Hub, UC Berkeley.
\end{acks}

\bibliographystyle{ACM-Reference-Format}
\bibliography{refs}

\appendix
\input{A-appendix}

\input{B-App}

\end{document}

%% file: 0-abstract.tex
\begin{abstract}
		Previous works in the differential privacy literature that allow users to choose their privacy levels typically operate under the heterogeneous differential privacy (HDP) framework with the simplifying assumption that user data and privacy levels are not correlated. Firstly, we demonstrate that the standard HDP framework falls short when user data and privacy demands are allowed to be correlated. Secondly, to address this shortcoming, we propose an alternate  framework, Add‐remove Heterogeneous Differential Privacy (AHDP), that jointly accounts for user data  and privacy preference.
        We show that AHDP is robust to possible correlations between data and privacy. 
        Thirdly, we formalize the guarantees of the proposed AHDP framework through an operational hypothesis testing perspective. 
	The hypothesis testing setup may be of independent interest in analyzing other privacy frameworks as well.
	Fourthly, we show that there exists non-trivial AHDP mechanisms that notably do not require prior knowledge of the data-privacy correlations.
    We propose some such mechanisms and apply them to core statistical tasks such as mean estimation, frequency estimation, and linear regression.
	The proposed mechanisms are simple to implement with minimal assumptions and modeling requirements, making them attractive for real-world use.
	 Finally, we empirically evaluate proposed AHDP mechanisms, highlighting their trade-offs using LLM-generated synthetic datasets, which we release for future research.
\end{abstract}

\keywords{Differential Privacy, Personalized Privacy, Correlations in Privacy and Data}

%% file: 1-intro.tex
\section{Introduction}
\label{sec:intro}

Widespread and cheap availability of internet, compute and storage, has led to incessant user data collection by various services.
Data is a valuable commodity for uses such as targeted advertisements and determining insurance premiums, and thus, a market for data has been established \cite{Data17,Data21}.
A rise in public awareness has led to a growing demand for privacy, as reflected in legislative actions like the GDPR in Europe \cite{GDPR} and the California Consumer Privacy Act (CCPA) \cite{CCPA}.
Over the past couple of decades the framework of Differential Privacy (DP) \cite{DW06,Dwork06} has become the gold standard for analysis in academia.
DP has been deployed in several real-world tasks by organizations like the U.S. Census \cite{Erlin14}, Apple \cite{2017LearningWP}, and Google \cite{Abowd18}.
See \citet{Desfontaines2020RealWorldDP} for a selected list of major real-world deployments of DP.

A significant portion of the research on DP has focused on the case where every user is provided with the same level of privacy.
However, studies have found that different people tend to have different privacy `sensitivities' \cite{Ackerman99}.
For example. In \citet{Acquisti2005privacy}, the authors demonstrate a heterogeneity in privacy concerns among survey users.
Sometimes, this heterogeneity of privacy concern is naturally observed when users of a service can choose to opt-into a feature at some privacy cost.
For example, in an old Facebook feature, users could share their current location in order to be notified when their friends are nearby \cite{Li12,Chaudhuri24}.
Thus, heterogeneity in privacy demand arises naturally in real-world applications.
The line of work on mechanism design for data auctions graciously embraces the idea of providing different levels of privacy to different users \cite{Asu22,Pandey23,Anjarlekar23,Sen24,Kang23,Fallah24,Cummings23}.
However, a rather uncomfortable aspect that researchers have shied away from when dealing with heterogeneous privacy demand is the question of dealing with correlations of user data and user privacy demand.
For example, if users are sharing their body weights with a service provider, then it is conceivable that a user with extremely high or low body weight may demand more privacy than others.
For research areas like mechanism design, this critical question has remained unaddressed.

We focus on this regime of heterogeneous DP where user data and privacy demand can be correlated. 
We consider the central-DP setup (see \cref{fig:CentralDP}) where users send their true data and privacy preferences to the central server, and expect the server to publish data that adhere to their proposed privacy constraint under an agreed upon privacy framework.
We highlight some fundamental limitations of the heterogeneous DP framework used in the literature under data-privacy correlations and show that an alternate framework resolves some of these issues.

\begin{figure}[H]
    \centering
    \includegraphics[width=0.7\linewidth]{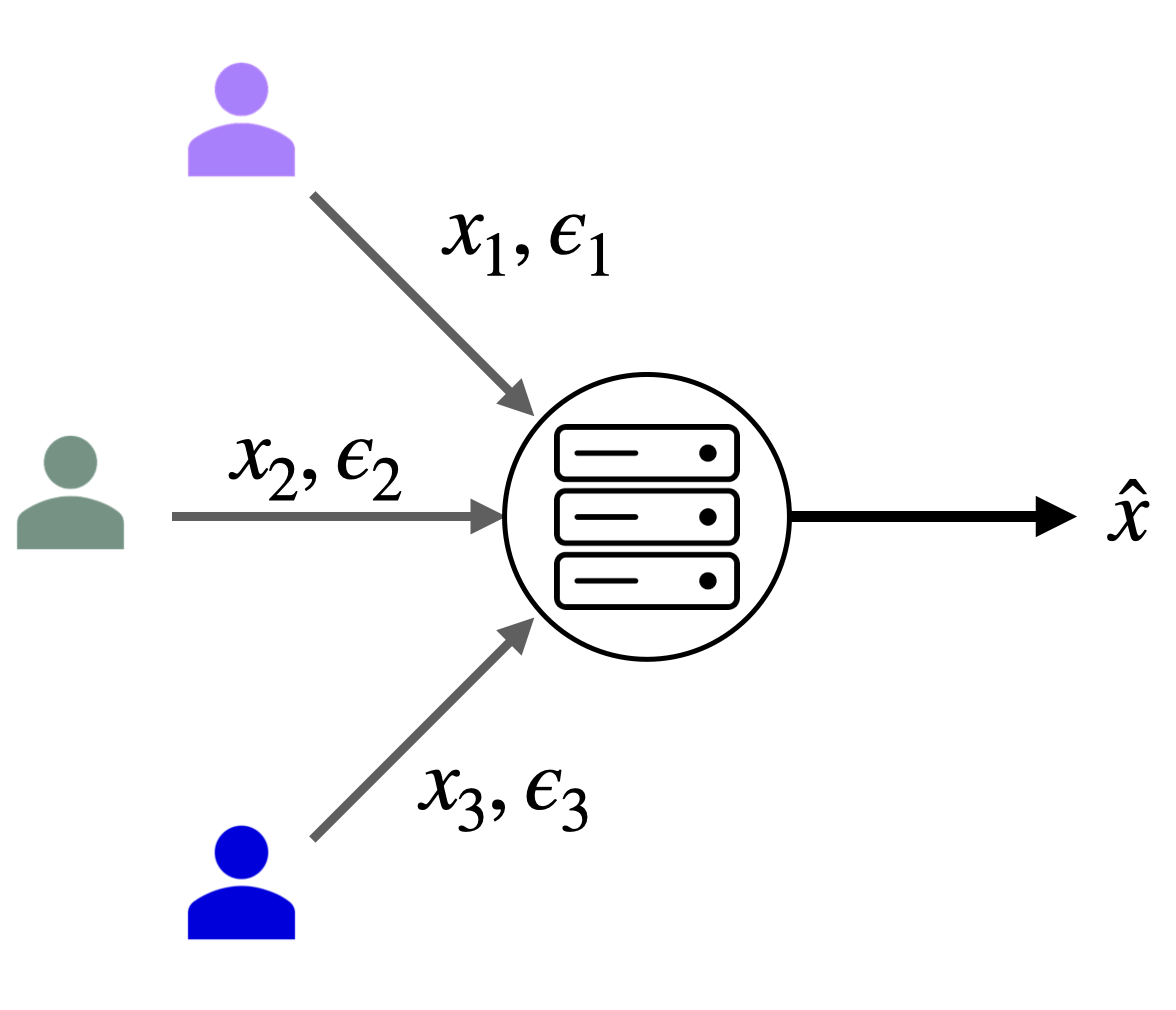}
    \caption{Central-DP: users send their true data, $x$, and privacy demand, $\epsilon$, to a central server.}
    \label{fig:CentralDP}
\end{figure}

\subsection{Our Contributions}

We focus on central-DP  in this work; in central-DP \cite{DW06}, users send their  data to a trusted server which promises to use a private mechanism on the user data.

Recall that in homogeneous $\epsilon$-DP,   `neighboring' datasets are only distinguishable up to a level determined by $\epsilon$, where smaller $\epsilon$ corresponds to more privacy.
Under the add-remove model of neighbors, two datasets are considered neighbors if one of the datasets can be  obtained by removing some datapoint from the other.
In the swap model of neighbors, two datasets are considered neighbors if they have exactly one differing datapoint.
In Heterogeneous DP (HDP), under data-privacy correlations, we show that mechanisms satisfying the standard definition of HDP used in the literature fall short of what one might expect from a private mechanism.
This pitfall is ultimately due to an unsuitable notion of neighbors implied by the standard HDP definition; that is, two datasets are considered neighbors if one user's data is changed but their privacy demand remains the same.

To remedy this, we consider an alternate definition of HDP in conjunction with a definition of `correlation' of data and privacy.
We term the proposed framework as Add-remove Heterogeneous Differential Privacy (AHDP). 
Crucially, the AHDP definition depends on the set of possible values that a tuple -- user data and user privacy demand -- can take, similar to how the homogeneous DP definition depends on the domain of data.
The domain of this tuple defines the the correlation  between the user data and user privacy demand.
Roughly speaking, the AHDP framework ensures that a mechanism's output doesn't help distinguish a particular dataset from another one obtained by appending to it  a new user's data and privacy demand.
The AHDP framework mitigates the  problem raised with the standard HDP definition.

Since AHDP is a new privacy definition, it is important to understand the what the `privacy' guarantees really mean, and we believe this is best understood through the lens of hypothesis testing.
We introduce a framework for analyzing the privacy of any privacy definition via hypothesis testing in \cref{sec:HT}.
In a nutshell, the framework characterizes the `identifiability' of a dataset, from a set of possible datasets, based on a mechanism output that satisfies the required privacy definition.
This is multiple hypothesis testing problem and the set of possible datasets, which form the set of hypotheses, can be thought of as the threat model.
We analyze AHDP and standard DP in this hypothesis testing framework for comparison.
In addition, we demonstrate how this framework can be used to demonstrate that standard HDP fails to provide meaningful privacy guarantees under correlations.

Although the AHDP framework relies on the set of possible tuples over user data and privacy demand, we demonstrate that mechanisms can be designed without requiring knowledge of these `correlations'. 
These mechanisms, which we call universal AHDP mechanisms, satisfy the AHDP property for any underlying domain of user data and privacy demand.
This key property allows universal AHDP mechanisms to be readily deployed in practical settings without the need for explicit correlation modeling.
We demonstrate some universal AHDP mechanisms for common tasks like sum, count, mean estimation, and linear regression that are simple to implement.
Some experiments  for  universal AHDP mechanisms are also presented to investigate which scheme works best in practice. \\

To summarize our contributions, we 
\begin{enumerate}
    \item Show that the standard HDP definition does not provide meaningful privacy guarantees when data and privacy can be correlated.
    \item Introduce a more suitable privacy definition, termed AHDP, that can alleviate some of the aforementioned issues. %
    \item Propose a general hypothesis test-based method to evaluate the privacy guarantees associated with different frameworks and present the AHDP framework's guarantees under this method.
    \item Show that there exist mechanisms that are AHDP without requiring prior knowledge of  possible data-privacy correlations, making them useful in practice. We propose such mechanisms for practically relevant problems like mean estimation,  frequency estimation, and linear regression.
    \item We experimentally compare some of these methods to demonstrate their utility.
\end{enumerate}

We emphasize that this work is to be seen as a beginning step in dealing with heterogeneity in DP under data-privacy correlations -- a question long ignored by the academic community.
When considering heterogeneous privacy, such correlations must be taken into account and the corresponding challenges that stem from it.
A better grasp of these challenges can help better understanding the true trade-offs associated with homogeneous and heterogeneous DP frameworks.
We envision this line to work to be most relevant to the mechanism design community due to the inherent nature of heterogeneity rooted in the modeling assumptions present in most works.

\paragraph{Organization:} \cref{sec:PD} introduces the necessary notations, background on homogeneous DP, heterogeneous DP, and the issues associated with heterogeneous DP used in literature.
\cref{sec:AHDP} introduces our proposed AHDP framework and some basic properties of AHDP.
We present a study of threat model and operational understanding of the privacy afforded by AHDP in \cref{sec:HT}.
We present some AHDP conforming mechanisms in \cref{sec:algo} and \cref{sec:exp} contains experiments on comparing different AHDP mechanisms.
We offer some concluding thoughts and future directions in \cref{sec:con}.

\subsection{Related Work}

In differential privacy (DP), there are two parallel regimes known as central-DP \cite{DW06} and local-DP \cite{Kasiviswanathan11}.
We focus on the central-DP regime in this work.
Several modifications of DP have been proposed, see \citet{desfontaines2019sok} for a selected list of frameworks.

A form of heterogeneity in privacy demand considered in recent literature is the case of a (homogeneous) private dataset supplemented with some publicly available dataset  \cite{Bassily20,Bassily20-learn,Liu21,Alon19,Nandi20,Kairouz21,Amid22,Wang19,Kamath22}.
While employed in deep learning training and alignment, a critique of such paradigm presented by \citet{Tramer22} is that publicly available is not necessarily data that requires no privacy.

More generally, some early works \cite{Alaggan17,Ghosh11,Jorg15}  consider the question of Heterogeneous DP (HDP), also known as Personalized DP, and use the standard HDP definition (\cref{def:hetDP}) \footnote{\citet{Jorg15} employ a definition similar to the one we use, see \cref{sec:AHDP}.}.
The works of \citet{Alaggan17} and \citet{Jorg15} present versatile algorithms for HDP but they are task agnostic so they do not necessarily perform at par with other techniques tailored to specific tasks \cite{journal23}.
HDP has been considered for tasks like mean estimation \cite{Asu22,isit-paper,journal23,Cummings23,Ferrando21}, frequency estimation \cite{Chaudhuri24}, ridge regression \cite{Acharya24}, and logistic regression \cite{Anjarlekar23}.
A popular assumption in these works is that the user data does not influence the user privacy demand, i.e., for a given privacy demand from users, the user data is i.i.d. from some unknown distribution.
However, this assumption is not justified by any particular work and can be seen as an approximation to bypass technical challenges.

HDP is a natural topic of research in mechanism design literature where the focus is on coming up with and analyzing markets for data \cite{Asu22,Pandey23,Anjarlekar23,Sen24,Kang23,Fallah24,Cummings23}.
Users are modeled to have some inherent and differing `privacy sensitivity' and the server aims to maximize some accuracy metric in estimation task.
The server comes up with a choice of privacy to be provided to each user and a corresponding payment.
Usually, the privacy sensitivity and the payments are not assumed to be requiring privacy.

There exists a separate line of research that focuses on the effect of correlations between different features or datapoints on the privacy loss \cite{aliakbarpourenhancing,zhang2022correlated,chong2024may}.
In the past, \citet{Chaudhuri24} and \citet{Ghosh11} have considered the question of correlation between user data and privacy.
\citet{Chaudhuri24} model correlations between data and privacy in a worst-case setting for statistical estimation tasks like mean and frequency estimation.
However, they assume that the privacy demanded by the users is public so they can use the standard HDP definition.
\citet{Ghosh11} show a negative result that a dominant strategy incentive compatible mechanism for selling privacy at an auction is not possible in a meaningful way where privacy is provided to both user's data and privacy sensitivity, however, as we point out, the standard HDP definition employed by them is not suitable under correlations.

Hypothesis testing under differential privacy has been studied before in literature \cite{Wasserman10,Kairouz15}.
In fact, the $f$-DP framework proposed by \citet{dong2022gaussian} is defined in terms of hypothesis testing; they argue that differential privacy is most naturally defined in terms of hypothesis testing.
Past works in membership inference attack have also used DP as a solution \cite{ye2022one} but these are predominantly homogeneous DP.
Other notions of `information leakage' have also been considered \cite{alvim2012differential, asoodeh2014notes} but these works typically do not focus on the hypothesis testing perspective.

\citet{Jorg15} also consider the add-remove notion of neighbors in HDP but they do not focus on correlations that may exist between data and privacy demand.
Considering correlations between data and privacy demand requires us to first define what these correlations mean.
We model these correlations as domain of possible values the user data and privacy demand can take, similar to how traditional DP definition depends on the domain of the data (see \cref{def:newDP}).
We show that the Sampling Mechanism proposed by \citet{Jorg15} satisfies our proposed AHDP privacy framework and does not need to be aware of the domain of correlations that exist in the data, i.e., it is an universal-AHDP mechanism.

%% file: 2-pd.tex
\section{Background and Problem Definition} \label{sec:PD}

\subsection{Notation}
With some abuse of notation, we represent $\bbR \cup \{ \infty \}$ by $\bbR$.
We denote positive real numbers by $\bbR_{> 0}$ and non-negative reals by $\bbR_{\geq 0}$, similarly define $\bbZ_{> 0}$ and $\bbZ_{\geq 0}$ for integers.
The notation $[n]$ refers to the set $\{1,2,\ldots,n\}$.
The probability simplex in $n$-dimensions is represented by $\Delta_n$.
For a set $\cX$, if $\vecx$ denotes a sequence of $n$ elements such that $\vecx \in \cX^n$, we refer to the $i$-th element of $\vecx$ by $x_i$.
For a set $\cX$, we represent a multiset on $\cX$ as the order pair $D = (\cX,h)$ where $h: \cX \to \bbZ_{\geq 0}$ counts the number of occurrences of each element of $\cX$ in $D$.
Denote the set of all multisets with countable elements by $\cS(\cX) := \{(\cX,h)| h:\cX \to \bbZ_{\geq 0}, \exists n \geq 0, \sum_{x \in \cX}h(x) = n \}$.
We occasionally use the notation $\mset{\cdot}$ to represent in multiset in an unordered manner.
Standard multiset notations are used \cite{Costa21}, see \cref{A:multisets} for definitions of operators on multisets used in this work.

With some abuse of notation, for a randomized algorithm $M$ mapping inputs to a probability distribution on output space, $M(\cdot)$ will interchangeably be used to refer to the output distribution or a sample from it.
 $|\log \frac{0}{0}|$ should be interpreted as $0$ and $|\log \frac{c}{0}| = \infty$ for any $c>0$.
 The notations $\lesssim$ and $\simeq$ denote inequality and equality that hold up to a universal multiplicative constant.
The notation $\mathsf{Clip}_{[a,b]}(x)$ for $a \leq b$ is defined as  $\min\{b,\max\{a,x\}\}$.

\subsection{Background on Differential Privacy}

In this work, we restrict ourselves to the Central-DP regime where users send their true data to the server and trust the server to adhere to a private algorithm.

We begin with the well-known definition of Differential Privacy (DP) in \cref{def:pureDP} \cite{DW06,Dwork06}, adapted from \cite{Kulesza24}.
In this work, we shall have inputs to the mechanism as multisets, i.e., an order pair of elements from a domain $\cX$ and a corresponding count of each element.

\begin{definition}[Differential Privacy] \label{def:pureDP}
A randomized algorithm $M: \cS(\cX) \to \cY$ is said to be $\epsilon$-DP for $\epsilon \in \bbR_{\geq 0}$ if for all measurable sets $S \subseteq \cY$,
\begin{equation} \label{eq:pureDP-def}
    \left| \log  \frac{\Pr\{M(D) \in S\}}{\Pr\{M(D') \in S\}} \right| \leq \epsilon,
\end{equation}
where  $D,D'$ are `neighboring' datasets.
\end{definition}

There are two main types of neighbors that are considered in literature, namely, the swap model of neighbors and the add-remove model of neighbors.
We outline the two notions next.
\begin{enumerate}
    \item \textbf{Swap model:} In the swap model, datasets $D$ and $D'$  are neighbors if they are of the same size with at most one element being different , i.e., $|D| = |D'| > 0$ and $|D - D'| = |D' - D| = 1$.
    \item \textbf{Add-remove model:} In the add-remove model, datasets $D$ and $D'$ are neighbors if one of them can be obtained from the other by deleting or adding an element. In other words, $|D - D'| + |D' -  D| = 1$.
\end{enumerate}

An issue with considering the input to the mechanism as a multiset is that it restricts the mechanism to be order invariant for the presented dataset.
This treatment is similar to the one presented in \cite{Dwork14Alg} but other definitions of DP that need not be order invariant have been considered in literature.
In fact, order invariant mechanisms have been shown to boost privacy in an appropriate sense, see \citet{Erlingsson19} for reference.

In the framework of Differential Privacy, lower $\epsilon$ means higher privacy since a lower value of $\epsilon$ constrains the algorithm from changing the output distribution as the inputs are varied.
The add-remove model protects the size of the dataset as well as the values, while the swap model only protects the values present in the dataset and not the size of it.
Further, an add-remove $\epsilon$-DP algorithm can be shown to be a swap $2\epsilon$-DP algorithm but a transformation from swap model to add-remove model does not hold in general. 
Readers can refer to \citet{Kulesza24} for more discussion.

We take a moment to note that by applying \cref{def:pureDP} multiple times (known as Basic Composition), \cref{def:pureDP} implies
\begin{equation} \label{eq:altDP-def}
    \left| \log  \frac{\Pr\{M(D) \in S\}}{\Pr\{M(D') \in S\}} \right| \leq \epsilon d(D,D'),
\end{equation}
where $d(D,D')$ `measures' the distance between datasets $D$ and $D'$.
In the swap model, we have
\begin{equation}
d_{\mathsf{swap}}(D,D') = \begin{cases}
    |D - D'| & \text{ if $|D| = |D'|$}, \\
    \infty & \text{else,}
\end{cases}    
\end{equation}
and in the add-remove model, we have
\begin{equation}
d_{\mathsf{ar}}(D,D') = 
    |D - D'| + |D'-D|.
\end{equation}
It is straightforward to see that \cref{eq:altDP-def} also implies \cref{def:pureDP}, making them equivalent.
$\epsilon$-DP ensures that two datasets that are neighbors can only be `distinguished' up to a level decided by $\epsilon$.

One way to view this DP condition is to construct a graph $\cG = (V,E)$ with the vertices $V = \cS(\cX)$ being the set of datasets and edges between vertices $D$ and $D'$ such that $d(D,D') = 1$.
In this interpretation, DP-based privacy between two datasets degrade as the datasets get further apart in the graph, and DP provides no guarantees for datasets that are not connected in the graph $\cG$.
The latter point does not arise in the add-remove model where any two datasets are always connected in the graph, unlike the swap model where datasets of different sizes are disconnected.
Other modifications of DP consider different model of neighbors and edge weights, see \citet{desfontaines2019sok} for more details.

\subsection{Background on Heterogeneous Differential Privacy}

Heterogeneous Differential Privacy (HDP) permits users to have different  privacy requirements.
The standard definition for HDP considered in literature uses the swap model for neighbors \cite{Alaggan17,Asu22,Chaudhuri24}.
Unlike the previous section, we present this definition in \cref{def:hetDP} with the input to the algorithm having domain $\cX^n$ instead of a multiset for ease of comparison with previous works.
We use the notation $\vecx \sim_i \vecx'$ to denote $x_j = x'_j \ \forall j \in [n] \setminus \{i\}$. \\

\begin{definition}[Heterogeneous Differential Privacy - Conventional Definition] \label{def:hetDP}
A randomized algorithm $M: \cX^n \to \cY$ is said to be $\bm\epsilon$-DP for $\bm\epsilon \in \bbR_{\geq 0}^n$ if $\forall i \in [n]$, for all measurable sets $S \subseteq \cY$,
\begin{equation} \label{eq:hetDP-def}
    \left| \log  \frac{\Pr\{M(\vecx) \in S\}}{\Pr\{M(\vecx') \in S\}} \right| \leq \epsilon_i,
\end{equation}
where  $\vecx,\vecx' \in \cX^n$ satisfy $\vecx \sim_i \vecx'$.
\end{definition}

In this model, the users communicate their privacy preferences to the server as well as their data.
For example, user $i$ demands for an $\epsilon_i$ level of privacy and has data $x_i$.
In the few works on HDP in literature, a central assumption in this line of work is to assume that the data $x_i$ and privacy demand $\epsilon_i$ are `independent' in an appropriate sense.
For example, some works consider the setting that for a given privacy demand $\bm\epsilon$, the data $\{ x_i\}_{i=1}^n$ is i.i.d. from some unknown distribution \cite{Asu22,journal23,isit-paper}.

\subsection{Issues with Previous HDP Formulation Under Data and Privacy Correlation}
\label{sec:Issues}

We claim that \cref{def:hetDP} does not adequately protect privacy when we allow the possibility of correlations between data and privacy.
Assume that an adversary who tries to infer private data from the output of a mechanism is aware of the correlations between the data and privacy demand.
The exact threat model is discussed later but for now, assume that adversary's knowledge of correlations means that it knows the joint data and privacy distribution $P(x,\epsilon)$.

In such scenario, a private mechanism must also protect the privacy demand of the users to the extent of correlations present in data and privacy.
In \cref{ex:swap-dp}, we provide an example to outline an issue with \cref{def:hetDP}.

Let $\cL(\lambda)$ denote a sample from the Laplace distribution with pdf given by $p(x;\lambda) = \frac{1}{2\lambda} e^{-|x|/\lambda}$.

\begin{example}[Issue with \cref{def:hetDP}] \label{ex:swap-dp}
    Let there be $n$ users with $\cX = \{0,1\}$, indicating some sensitive information and let $\vecx$ be the dataset.
    As outlined in \citet{Asu22}, all mechanisms of the form $\inprod{\vecw}{\vecx} + \cL(\|\frac{\vecw}{\bm\epsilon}\|_{\infty})$ satisfy $\bm\epsilon$-DP in the sense of \cref{def:hetDP} for $\vecw \in \Delta_n$, where $\frac{\vecw}{\bm\epsilon}$ denotes element-wise division.
    For illustration, consider the particular mechanism with $\vecw \propto \bm\epsilon$, i.e., $M(\vecx) = \mathsf{Clip}_{[-1,3]}\left(\frac{\inprod{\bm\epsilon}{\vecx}}{\|\bm\epsilon\|_1} + \cL(\frac{1}{\|\bm\epsilon\|_1})\right)$.
    $M(\cdot)$ remains $\bm\epsilon$-DP by post-processing property.
    
    Consider the coupling (joint-distribution) between data and privacy such that $(x,\epsilon) \in \{(1,0),(0,\infty) \}$.
    Let $\vecx$ be the dataset with all $n$ users having data $x=1$ and privacy $\epsilon=0$.
    Then $M(\vecx)$ essentially follows a $4\times$Bern$(0.5) - 1$ distribution. 
    Under mechanism $M$, any other dataset of size $n$ will have a deterministic output of $1$.
    Thus, by observing a sample from $M(\vecx)$, an adversary can figure out the dataset $\vecx$, even though all the users asked for full privacy in the considered dataset $\vecx$.
\end{example}

Thus, the HDP definition of \cref{def:hetDP} does not intuitively serve its supposed purposed of providing privacy when there are correlations between the data and privacy and an adversary is aware of the correlations.
While \cref{ex:swap-dp} may seem contrived, it does highlight a valid shortcoming of the HDP definition.

A less contrived example is the following - the mechanism that deterministically outputs the privacy demand vector $\bm\epsilon$ satisfies the $\bm\epsilon$-DP definition in \cref{def:hetDP} since it does not output the data at all.
However, an adversary may be able to infer about the data of users asking extremely high-level of privacy from this mechanism, highlighting a privacy-leakage when data and privacy demand can be correlated.

Recall that the nature of DP under the notion of \cref{def:hetDP}, any neighbor of a dataset has the exact same privacy demand, but with possibly different user data.
Under the correlation model of \cref{ex:swap-dp}, the dataset considered has no neighbors, i.e., it is a disconnected node in the graph since no user can have different data but the same privacy demand.
Thus, it is not surprising that we can distinguish it from other datasets in the graph $\cG$, and this leads us to consider a better notion of neighbors instead of \cref{def:hetDP}.

%% file: 3-shdp.tex
\section{Add-remove Heterogeneous Differential Privacy (AHDP)} \label{sec:AHDP}

Recall that we operate in the Central-DP regime where the users send their true data and privacy demand to the central server.
Since we focus on the setting where data and privacy demand are allowed to be correlated, we emphasize that the {\em privacy demand is not public}.
Therefore, the mechanism cannot arbitrarily leak the privacy demand of the users.

We shall represent the elements of a dataset by tuples of the user data and privacy demand, and we shall switch back to representing datasets by multisets.
Let $\cW \subseteq \cX \times \bbR_{\geq 0}$ denote the set of all possible data and privacy demand in the problem considered.
$\cW$ can be thought of as all possible pairs of data and privacy that can be present in the problem under consideration - i.e., the correlations.
For example, $\cW = \{(1,0),(0,\infty)\}$ in \cref{ex:swap-dp}.
Thus, the set of all possible datasets is denoted by $\cS(\cW)$. 
For simplicity, we shall assume that $\cW$ is a finite set in the rest of this work.

We consider an alternate definition for HDP that alleviates the shortcomings of \cref{def:hetDP}.
In order to distinguish it from \cref{def:hetDP}, we shall refer to it as Add-remove Heterogeneous Differential Privacy (AHDP) from here on.

Consider a function $\alpha: \cW \to \bbZ_{\geq 0}$ and let $h_D$ denote the count function associated with a multiset $D$.
For $D,D' \in \cS(\cW)$, we define
\begin{equation}
d_{\alpha}(D,D') = \sum_{(x,\epsilon) \in \cW} \alpha(x,\epsilon) | h_{D}(x,\epsilon) - h_{D'}(x,\epsilon) |. 
\end{equation}
Note that $d_{\alpha}(D,D') = d_{\alpha}(D',D)$. 
The distance $d_{\alpha}(D,D')$ measures how far part the datasets $D$ and $D'$ are on the graph obtained by the add-remove model of neighbors while weighing the edges by the $\alpha$ function.
For clarity, we repeat the assumptions of our problem setup:
\begin{enumerate}
    \item we operate in the central-DP regime where users send their data and privacy demand tuple to the server,
    \item the user data and privacy demand tuple is only known to the server,
    \item there exists a domain $\cW \subseteq \cX \times \bbR_{\geq 0}$ of possible data and privacy demand, which serves as our model for correlations, and
    \item the adversary may be aware of joint distributions of user data-privacy tuple.
\end{enumerate}
\cref{def:newDP} defines the notion of $\alpha$-AHDP and $\cW$-AHDP.

\begin{definition}[AHDP] \label{def:newDP}
(A) A randomized algorithm $M: \cS(\cW) \to \cY$ is said to be $\alpha$-AHDP if for all measurable sets $S \subseteq \cY$,
\begin{equation} \label{eq:AHDP-def}
    \left| \log  \frac{\Pr\{M(D) \in S\}}{\Pr\{M(D') \in S\}} \right| \leq d_{\alpha}(D,D') ,
\end{equation}
where  $D,D' \in \cS(\cW)$. \\
(B) A mechanism $M:\cS(\cW) \times \cT \to \cY$ is said to be $\alpha$-AHDP if $M(\cdot,t)$ is $\alpha$-AHDP $\forall t \in \cT$. \\
(C) An $\alpha$-AHDP algorithm is said to be $\cW$-AHDP if $ \forall (x,\epsilon) \in \cW, \ \alpha(x,\epsilon) \leq \epsilon$. 
\end{definition}

$\alpha$-AHDP is equivalent to ensuring that $\forall D \in \cS(\cW)$ and $D' = D + \mset{(x,\epsilon)}$, we have for any measurable set $S$
$$\left|\log \frac{\Pr\{M(D) \in S\}}{\Pr\{M(D') \in S\}}\right| \leq \alpha(x,\epsilon),$$ 
Thus, $\alpha$ can be understood as the mapping from the data-privacy tuple to the level of privacy offered by the mechanism; this is inherently heterogeneous and can differ for each input tuple due to heterogeneity of privacy demands. 
In \cref{sec:HT}, we present an interpretation of the privacy guarantee in this context. 
Roughly speaking, AHDP prevents inference of individual user data demand to the extent allowed by the privacy level specified. Unlike HDP, AHDP prevents leakage due to correlations between user data and privacy demand. A more expanded exposition on interpreting the privacy afforded by AHDP can be found at the end of \cref{sec:HT}.

\cref{def:newDP}(B) defines the property of $\alpha$-AHDP when the mechanism can have auxiliary  input parameter $t \in \cT$.
Part (C) is a Boolean classification - the $\alpha$-AHDP mechanism is meaningful only if it respects the user privacy demand, i.e., if $\alpha(x,\epsilon) \leq \epsilon$ $\forall (x,\epsilon) \in \cW$.
\citet{Jorg15} also consider the add-remove model of neighbors in their HDP definition similar to \cref{def:newDP}, however their focus was on the regime that data and privacy are not correlated.
We demonstrate in \cref{sec:HT} that AHDP can provide meaningful privacy guarantees even under correlation between data and privacy.

\cref{ex:ahdp} demonstrates how homogeneous DP can be recovered as a special case of AHDP and \cref{ex:ahdp-via-homo} demonstrates how homogeneous DP can be used to achieve AHDP.
We show \cref{sec:algo} that we can construct better mechanisms than \cref{ex:ahdp-via-homo}.

\begin{example}[Reduction to homogeneous DP] \label{ex:ahdp}
As an example, if $\cW_h = \cX \times \{\epsilon_h\}$, then a homogeneous $\epsilon'$-DP mechanism (in add-remove model) satisfies $\alpha_h$-AHDP where $\alpha_h(x,\epsilon_h) = \epsilon'$.
We can certify this mechanism to be $\cW_h$-AHDP iff $\epsilon' \leq \epsilon_h$.
\end{example}

\begin{example}[AHDP via homogeneous DP] \label{ex:ahdp-via-homo}
Suppose $\epsilon_l(\cW) = \min_{(x,\epsilon) \in \cW} \epsilon$ such that $\epsilon_l > 0$, then the mechanism providing homogeneous $\epsilon_l$-DP privacy is $\cW$-AHDP.
In other words, the mechanism providing the uniform level of privacy corresponding to the highest privacy demand to all the users is $\cW$-AHDP. However, this mechanism is only meaningful if $\epsilon_l(\cW)>0$.
\end{example}

\begin{remark}[Local-DP Implementation]
	One might be tempted to move to the local-DP regime where users perturb their data based on their privacy requirement and send just the perturbed data to the server, without revealing their privacy parameter.
	However, this implementation raises the question of how the server decides on which user data is useful (low privacy)?
	It is possible that a lot of users have high privacy preference and the data with the server is effectively just noise, but the server may not realize it.
    In other words, such implementations make it difficult to judge the quality and utility of the data for downstream tasks.
\end{remark}

As we show next, we can compose different mechanisms $\alpha_1$-AHDP and $\alpha_2$-AHDP to get an $(\alpha_1 + \alpha_2)$-AHDP mechanism, which can then be used to know if the composition is $\cW$-AHDP. 
We cannot directly compose two $\cW$-AHDP mechanisms to obtain a $\cW$-AHDP mechanism.
In other words, the $\cW$-AHDP property is certificate to be used on the final deployed mechanism while the $\alpha$-AHDP property is useful for design and analysis of algorithms.

\subsection{Composition and Post-Processing for AHDP} 

We show that basic composition and post-processing holds for AHDP; both are powerful tools used in design and analysis of privacy of algorithms.
\cref{prop:BC} states that if there are $k$ mechanisms $M_1,\ldots,M_k$ which are $\alpha_1,\ldots,\alpha_k$-AHDP respectively, then sequentially composing them is going to be $\sum_{i=1}^k\alpha_i$-AHDP.

\begin{restatable}{proposition}{BasicCom} \label{prop:BC}
    If mechanisms $M_i:\cS(\cW) \bigtimes_{j < i} \cY_{j} \to \cY_i$ are $\alpha_i$-AHDP in the sense of \cref{def:newDP}(B) and independent, then
    $M: \cS(\cW) \to \bigtimes_{i=1}^k\cY_i$ given by $M(D) = (M_1(D_1),M_2(D_2),\ldots,M_k(D_k))$ is $\left(\sum_{i=1}^k\alpha_i\right)$-AHDP, where $D_1 = D$ and $D_i = (D_{i-1},M_{i-1}(D_{i-1}))$ for $i > 1$.
\end{restatable}

The proof can be found in \cref{A:comp}.
\cref{prop:pp} shows that applying arbitrary transformations on the output of an $\alpha$-AHDP mechanism cannot degrade privacy. The proof is presented in \cref{A:pp}.

\begin{restatable}{proposition}{PostProc} \label{prop:pp}
    Let $M:\cS(\cW) \to \cY$ be an $\alpha$-AHDP mechanism and $K:\mathcal{Y} \to \mathcal{Z}$ be a measurable function. Then, $K\circ M$ is also $\alpha$-AHDP.  
\end{restatable}

%% file: 3.5-Privacy.tex
\section{Characterizing the Privacy Guarantee} \label{sec:HT}

We present an operational perspective on the privacy guarantees afforded by AHDP and the associated threat model in this section.
We begin by a simple hypothesis testing setup.
Suppose an adversary observes the output $Y$ of an AHDP algorithm and aims to identify the underlying dataset from two possible hypotheses.
Let $M$ be an $\alpha$-AHDP algorithm in the sense of \cref{def:newDP} and let $D,D' \in \cS(\cW)$ be the two hypothesis that the adversary wants to test against using the output $Y$. 
We set up the hypothesis testing problem with
\begin{align}
    H_0 &: Y \sim M(D), \\
    H_1 &: Y \sim M(D').
\end{align}
For any test, let the rejection region for the null be represented by measurable set $R \subseteq \cY$ and let the corresponding type I and type II error rates be represented by $e_1(R,D,D')$ and $e_2(R,D,D')$ respectively.
Note that $e_2(R,D,D') = e_1(\cY \setminus R,D',D)$.
Then, similar to the results known in homogeneous DP \citep{Kairouz15,Wasserman10}, we characterize the type I and type II error rates in \cref{thm:HT}.
The proof can be found in \cref{A:HT}.

\begin{restatable}{proposition}{thmHT} \label{thm:HT}
    A mechanism $M$ is $\alpha$-AHDP iff for any datasets $D,D' \in \cS(\cW)$, for any measurable rejection region $R \subseteq \cY$, we have
    \begin{equation}
        e_1(R,D,D') + e^{d_{\alpha}(D,D')}e_2(R,D,D') \geq 1.
    \end{equation}
\end{restatable}

A reasonable point to ponder upon is on how to interpret the hypothesis testing result of \cref{thm:HT}.
\textit{Presented set up inherently assumes that one of the two hypotheses is true but how does the adversary come up with the two hypotheses and what does \cref{thm:HT} mean for the adversary?}
We a framework for understanding the privacy guarantees next, which can be applied more generally to the privacy literature.

\subsection{A Framework for Privacy Loss}

We begin by focusing on homogeneous DP in the swap and add-remove model of neighbors.
In order to interpret the privacy guarantee provided by differentially private mechanisms, we introduce a notion of `adversarial power' to capture leakage of information.
Adversarial power is defined in \cref{def:capHomo} for homogeneous DP mechanisms.

\begin{definition}[Power for homogeneous DP] \label{def:capHomo}
    For a set of hypotheses $\cH = \{D_1,\ldots,D_n\} \subseteq \cS(\cX)$, called threat model, and a mechanism $M:\cS(\cX) \to \cY$, define the adversarial power to be
    \begin{equation}
        \cP(M,\cH) =   \max_{\psi:\cY \to \Delta_{\cH}} \min_{D \in \cH}  \Pr\{\psi(M(D)) = D\}.
    \end{equation}
\end{definition}

In other words, $\cP(M,\cH)$ measures how much an adversary can infer about underlying dataset given the output of a mechanism $M$.
The set of hypotheses $\cH$ can be understood as the threat model - the smaller this class is the more side-information the adversary had to reduce the possible datasets from $\cS(\cX)$ to $\cH$.
In fact, the power for any mechanism $M$ can be lower bounded by $\frac{1}{|\cH|}$ using uniformly random test $\psi$, showing power increases with a smaller set of hypothesis.
Subsequently, we  consider some natural threat models that can arise for the different DP frameworks. 

The minimum over $D \in \cH$ ensures that the power captures the worst-case situation for the adversary, i.e., the adversary cannot increase the power by using a test $\psi$ that always outputs a particular hypothesis.
It can equivalently viewed as the Bayes error over the worst-case prior.
Power is only high if the adversary can detect each hypothesis with high probability.

To understand the privacy afforded by $\epsilon$-DP, we need to consider two terms
\begin{itemize}
	\item upper bound: an upper bound on the power of an adversary for any $\epsilon$-DP mechanism,
	\item lower bound: existence of a mechanism $M'$ with suitable output space $\cY$ such that the power is at least some quantity.
\end{itemize}
In the following, we provide upper bounds on the power for any mechanism using \cref{thm:HT}-like results and provide lower bounds on power using exponential-based mechanism \cite{McSherry07} for different DP frameworks.

\paragraph{\textbf{Homogeneous DP (swap model):}}
In homogeneous DP, the dataset belongs to the set $\cS(\cX)$ and does not contain privacy demand.
Recall that in the swap DP model, the size of the dataset is fixed.
Consider the threat model where an adversary sees a subset of the data.
Let $D_o$ be the part of dataset observed by the adversary and let $D$ be the actual dataset.
In the swap model, the size of $D$ is not protected and hence, assume that $|D| = |D_o| + 1$ in the worst-case.
Equivalently, one can also think of all the other users to be colluding and sharing data with each other, and they are together trying to infer about the last user's data.
Concretely, assume
\begin{enumerate}
	\item the adversary sees a part of the dataset $D_o \in \cS(\cX)$, 
	\item $|D| - |D_o| = 1$,
	\item the adversary knows the mechanism $M$ used.
\end{enumerate} 
Adversary knowing the mechanism $M$ can be understood as the server announcing the mechanism it uses or the mechanism employed by the server being leaked to the adversary.
This assumption is similar in spirit to that used in cryptography where an eavesdropper is assumed to know the encryption-decryption scheme used.
Thus, the above interpretation naturally leads to a multiple hypothesis testing problem defined in \cref{def:capHomo}.
In the above context, threat model is $\cH =\bigcup_{x \in \cX} \{ D_o +\mset{x} \}$.
 \cref{claim:1} demonstrates the diminishing inference capability of the adversary as the domain size and the privacy increases.
 The proof can be found in \cref{sec:cap}.
 \begin{restatable}{proposition}{homoCap} \label{claim:1} Power under homogeneous $\epsilon$-DP (swap model) : 
 	 Suppose $|\cX|= k $, \\
    (A) for any $\epsilon$-DP mechanism $M$, $\cP(M,\cH) \leq \frac{1}{1+(k-1)e^{ - \epsilon}}$; \\
    (B) conversely, there exists an  $\epsilon$-DP mechanism $M'$ such that  $\cP(M',\cH) \geq  \frac{1}{1+(k-1)e^{ - \epsilon}}$.
\end{restatable}

\paragraph{\textbf{Homogeneous DP (add-remove model):}}
Now consider homogeneous DP under add-remove model. 
Suppose the adversary sees a part of the dataset $D_o $, this can be due to collusion of users or data leakage.
Let the actual dataset be $D$.
Assume the adversary knows that $|D| - |D_o| \leq t$, i.e., there are at most $t$ more datapoints in the dataset. 
In particular, the threat model assumes that adversary
\begin{enumerate}
	\item  sees a part of the dataset $D_o \in \cS(\cX)$,
	\item  has side-information that the actual dataset $D$ satisfies $|D| - |D_o| \leq t$ for $t \geq 1$,
	\item adversary knows the mechanism $M$ used.
\end{enumerate} 
Thus, the set of hypotheses is given by  $\cH_t = \bigcup_{i=0}^t \{D_o + D' | D' \in \cS(\cX), |D'| = i\}$.
In \cref{claim:2}, we provide upper and lower bounds on the power for the cases $t \to \infty$ and $t=1$.

 \begin{restatable}{proposition}{homoSwapCap} \label{claim:2} Power under homogeneous $\epsilon$-DP (add-remove model): Suppose $|\cX|= k $ then for any $\epsilon$-DP mechanism $M$ (add-remove model), \\
 	(A1) $\lim_{t \to \infty} \cP(M,\cH_t) \leq (1-e^{-\epsilon})^k$, \\
 	(A2) $\cP(M,\cH_1) \leq \frac{1}{1+ k e^{-\epsilon}}$.  \\
 	 Conversely, there exists  $\epsilon$-DP mechanisms $M', M''$ such that  \\
 	 	(B1) $\lim_{t \to \infty} \cP(M',\cH_t) \geq \frac{(1-e^{-\epsilon})^k}{(1+e^{-\epsilon})^k}$, \\
 	 (B2) $\cP(M'',\cH_1) \geq \frac{1}{1+ k e^{-\epsilon}}$. 
\end{restatable}

The case when $t=1$ is comparable to \cref{claim:1}, showing that under the rather strong assumption that adversary has side-information that there can only be at most one extra datapoint, add-remove and swap model of homogeneous DP offer comparable protection.
The case of $t \to \infty$ captures the case where adversary \textit{does not} have any side-information on the size of the dataset, and indeed, the power is lower than that of $t=1$. 
A plot for the upper bounds on the power is presented in \cref{fig:cap} as $k$ is varied; \cref{fig:cap-2} plots the upper bound varies as $\epsilon$ is changed.

\begin{figure}
    \centering
    \includegraphics[width=0.4\textwidth]{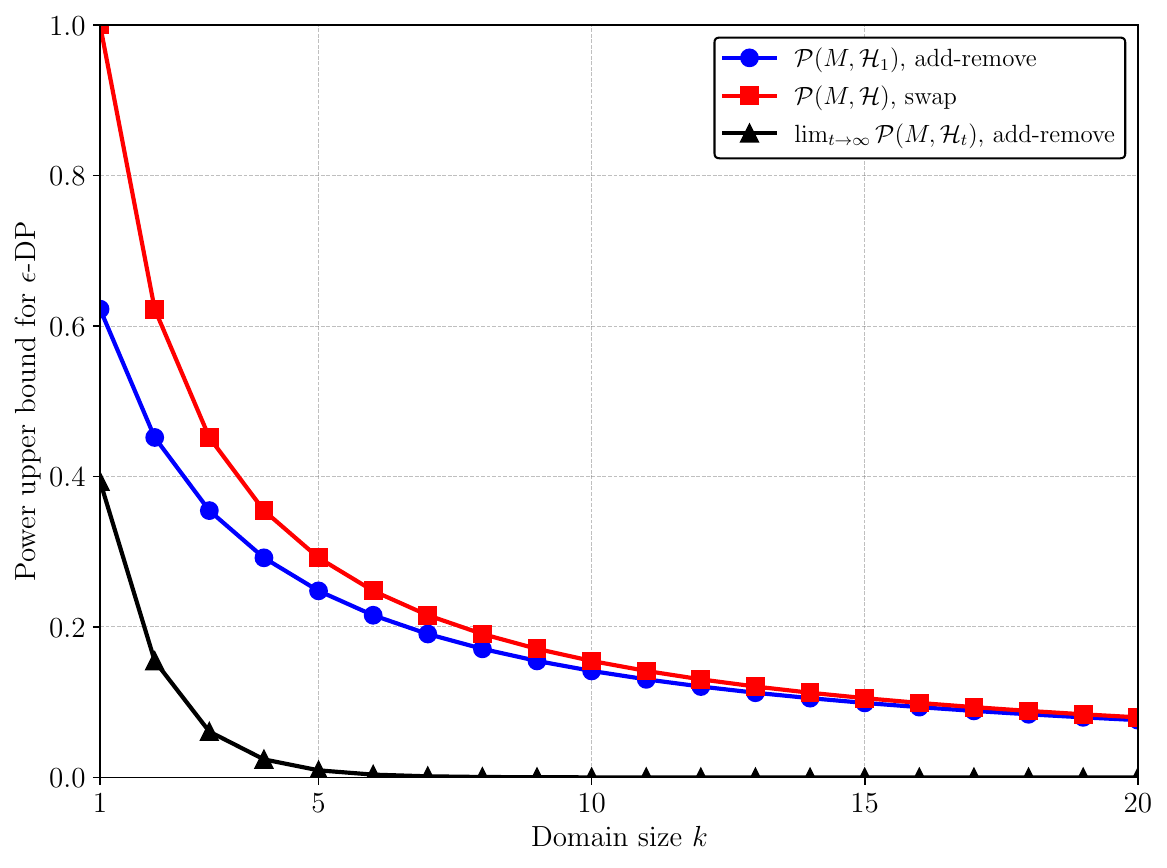}
    \caption{Plot of the upper bound for the power for $\cP(M,\cH)$ in \cref{claim:1} and $\cP(M,\cH_1)$, $\lim_{t \to \infty} \cP(M,\cH_t)$ in \cref{claim:2} as $k$ varies, keeping $\epsilon = 0.5$.}
    \Description{A plot of the three upper bounds as $k$.}
    \label{fig:cap}
\end{figure}

\begin{figure}
    \centering
    \includegraphics[width=0.4\textwidth]{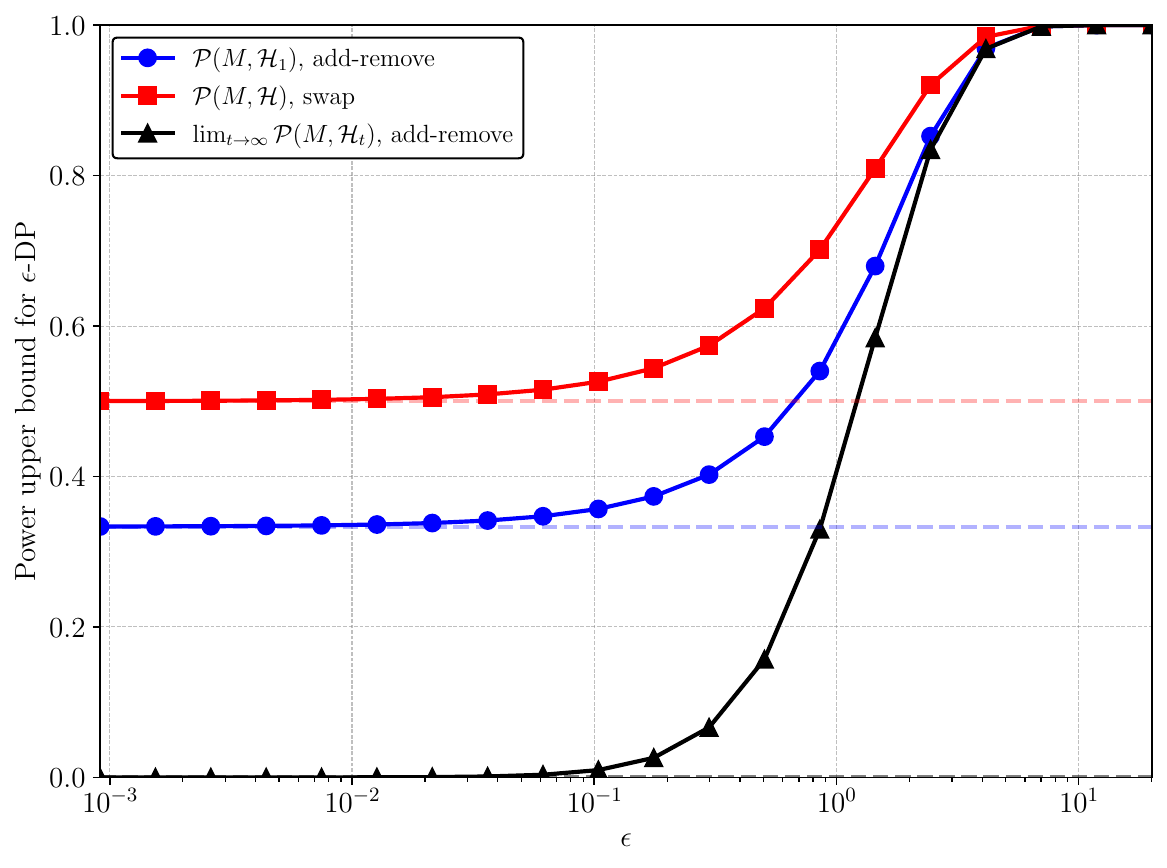}
    \caption{Plot of the upper bound for the power for $\cP(M,\cH)$ in \cref{claim:1} and $\cP(M,\cH_1)$, $\lim_{t \to \infty} \cP(M,\cH_t)$ in \cref{claim:2} as $\epsilon$ varies keeping $k=2$. The dashed horizontal line denote the trivial lower bound of $\frac{1}{|\cH|}$.}
    \Description{A plot of the three upper bounds as $\epsilon$ along with the trivial lower bounds as a dashed line.}
    \label{fig:cap-2}
\end{figure}

\subsection{Power In the AHDP Framework}

For the heterogeneous privacy regime, we need to amend the definition of power appropriately.
In particular, the adversary is considered successful in the hypothesis test if it can recover the underlying user data (not the tuple of data and privacy).
In other words, to capture the worst-case information leak, we consider the adversary's classifier is successful if it can infer the underlying user data.
In this setting, the dataset includes both the user data and user privacy demand.
For a dataset $D \in \cS(\cW)$, define the projection $\Pi_{\cX}(D) \in \cS(\cX)$ to be the marginal dataset on $\cX$, i.e., $h_{\Pi_{\cX}(D)}(x) = \sum_{(y,\epsilon) \in \cW: y = x} h_D(y,\epsilon) \ \ \forall x \in \cX$.
The power under heterogeneous privacy regime is \cref{def:cap-ahdp}.

\begin{definition}[Power under heterogeneity] \label{def:cap-ahdp}
	For the threat model $\cH = \{D_1,\ldots,D_n\}$, where $D_i \in \cS(\cW)$,
	let $\cH_{\cX} = \bigcup_{D \in \cH} \Pi_{\cX}(D) $.  
	For a mechanism $M:\cS(\cW) \to \cY$, define the power to be
	\begin{equation}
		\cP(M,\cH) = \max_{\psi:\cY \to \Delta_{\cH_{\cX}}}  \min_{D \in \cH} \Pr\{\psi(M(D)) = \Pi_{\cX}(D)\}.
	\end{equation}
\end{definition}

Similar to the homogeneous setting, the power is lower bounded by $\frac{1}{|\cH_{\cX}|}$, outlining that the power increases as the threat model becomes stronger ($\cH$ decreases).
The difference from the homogeneous setting is that the only the marginal of the dataset on $\cX$ matters to the adversary for inference.
Similar to homogeneous DP in the add-remove model, assume that the adversary observes $D_o \in \cS(\cW)$ by some mechanism such as collusion between users.
Concretely, in this threat model 
\begin{enumerate}
    \item adversary sees a part of the dataset $D_o \in \cS(\cW)$,
    \item knows the possible correlations $\cW$,
    \item adversary has side-information that the actual dataset $D$ satisfies $|D| - |D_o| \leq t$ for $t \geq 1$,
   	\item adversary knows the mechanism $M$ used.
\end{enumerate} 
Thus, for this threat model $\cH_t = \bigcup_{i=0}^t \{D_o + D' | D' \in \cS(\cW), |D'| = i\}$. 
\cref{claim:3} provides upper and lower bounds on the power for this threat model. \\

Let $$\epsilon_l(x) := \min \{ \epsilon| (x,\epsilon) \in \cW\},\ \cW_o = \{(x,\epsilon_l(x))| x \in \cX\},$$ i.e., $\epsilon_l(x)$ is the most privacy that can be demanded when the user data is $x$ in the set of correlations $\cW$ and $\cW_o$ is the subset of $\cW$ obtained by only considering these tuples.

\begin{restatable}{proposition}{hetCap} \label{claim:3} Power under AHDP: 
	For any $\cW$-AHDP mechanism $M$, \\
	(A1) $\lim_{t \to \infty} \cP(M,\cH_t) \leq \prod_{(x,\epsilon) \in \cW_o }(1-e^{-\epsilon})$, \\
    (A2) $\cP(M,\cH_1) \leq \frac{1}{1 + \sum_{(x,\epsilon) \in \cW_o }e^{-\epsilon}}$. \\
	Conversely,  there exists  $\cW$-AHDP mechanism $M', M''$ such that \\
	(B1) $\lim_{t \to \infty} \cP(M',\cH_t) \geq \prod_{(x,\epsilon) \in \cW_o } \frac{1-e^{-\epsilon}}{1+e^{-\epsilon}}$, \\
    (B2) $\cP(M,\cH_1) \geq \frac{1}{1 + \sum_{(x,\epsilon) \in \cW_o }e^{-\epsilon}}$.
\end{restatable}

Intuitively, \cref{claim:3} upper bound ensures that the power remains bounded if there is a possibility of a user demanding high privacy.
The set of datasets where $D_o$ has in addition some high-privacy demanding users is hard to distinguish from $D_o$ based on the mechanism output due to $\cW$-AHDP guarantee.
In a corner case such that $\cW$ has no possibility of high-privacy demanding users, the adversary can have higher power since privacy is not demanded by the users. \\

In the proposition, the expression is in terms of $\cW_o$ instead of  $\cW$ since the most private users dictate adversary's power.
For example, in the upper bound, one can consider the restricted hypotheses on those datasets where for a given user data $x$, the privacy demanded is the highest, i.e., the privacy is $\epsilon_l(x)$.
For the lower bound, one can construct an $\alpha$-AHDP mechanism such that $\alpha(x,\epsilon) = \epsilon_l(x)$, which also ensures the $\cW$-AHDP condition. \\

We also consider an alternate stronger threat model whose power is easier to interpret.
Assume
\begin{enumerate}
    \item adversary sees a part of the dataset $D_o \in \cS(\cW)$,
    \item adversary has side-information that the actual dataset $D$ is either  $D_o$ or $D_o + \mset{(x,\epsilon)}$ for some known $(x,\epsilon) \in \cW$,
   	\item adversary knows the mechanism $M$ used.
\end{enumerate} 
Thus, for this threat model $\cH(x,\epsilon) = \{D_o, D_o + \mset{(x,\epsilon)} \}$. 
\cref{claim:4} provides bounds on the power.

\begin{restatable}{proposition}{hetCapSec} \label{claim:4} Power under AHDP: 
	For any $\cW$-AHDP mechanism $M$, \\
	(A) $\cP(M,\cH(x,\epsilon)) \leq \frac{1}{1+e^{- \epsilon}}$. \\
	Conversely,  there exists a  $\cW$-AHDP mechanism $M'$ such that \\
	(B) $\cP(M',\cH(x,\epsilon)) \geq \frac{1}{1+e^{- \epsilon}}$. 
\end{restatable}

While this is easier to interpret, the stronger threat model is harder to motivate.
This is essentially the same threat model as that of the standard hypothesis testing results like \cref{thm:HT} assume.
The proof of \cref{claim:3} and \cref{claim:4} can be found in \cref{sec:cap}.
This threat model demonstrates that the adversary cannot distinguish between the two hypotheses $D_o$ and $D_o + \mset{(x,\epsilon)}$ if the unseen user has high privacy demand, ensuring that the privacy of high-privacy demanding users is preserved. 
So such guarantee is provided if the user has low privacy demand.

\begin{remark}[Power Depends on the Threat Model] \label{rem:power}
	The different adversarial power expressions that we show in this section depend on the threat model. As the adversary gets more side-information, its power will increase.
	To demonstrate this point, we consider the situation where $\cW = \{(x_1,0),(x_2,\infty)\}$ and suppose the adversary observes a part of the dataset $D_o$ and has the side-information that the actual dataset is either $D_o +\mset{(x_1,0)}$ or $D_o +\mset{(x_2,\infty)}$.
	In other words, the adversary knows the size of the dataset exactly.
	In this pathological threat model, since in the second possibility the user does not demand any privacy, the power can be $1$ for some $\cW$-AHDP mechanisms, highlighting that AHDP would not offer strong privacy guarantees in this threat model.
\end{remark}

\subsection{Interpreting privacy afforded by AHDP}

AHDP prevents inference of individual data to the extent allowed by the privacy level specified.
The exact inference power of the adversary depends on the side-information available to the adversary, as pointed out in \cref{rem:power}. 
Thus, standard HDP fails when data and privacy demand are correlated (see \cref{sec:Issues}), while AHDP remains robust under the same conditions as demonstrated in this section through the lens of hypothesis testing.

Analogous to standard DP \cite{desfontainesblog20200306,aliakbarpourenhancing}, AHDP can also be understood from a Bayesian setup. 
The following argument is completely analogous to the standard homogeneous DP literature.
Suppose there is an underlying distribution, a prior $\pi$, known to the adversary that generates datasets.
Suppose the adversary is trying to determine whether the underlying dataset $D$ is $D_o$ or $D_o + \mset{(x,\epsilon)}$?
The \textit{prior} odds is given by \[ r_{\mathsf{prior}} = \frac{\pi(D = D_o)}{\pi(D = D_o + \mset{(x,\epsilon)})}.\]
Now, on observing the output $Y$ from an AHDP mechanism $M(D)$, the adversary obtains \textit{posterior} odds of 
\begin{align}
    r_{\mathsf{posterior}} &= \frac{\pi(D = D_o|M(D) = Y)}{\pi(D = D_o + \mset{(x,\epsilon)}|M(D) = Y)}, \\
    &= \frac{\pi(D=D_o)}{\pi(D = D_o + \mset{(x,\epsilon)})} \frac{P(M(D_o) = Y)}{P(M(D_o+\mset{(x,\epsilon)}) = Y)}, \\
    &= r_{\mathsf{prior}}  \frac{P(M(D_o) = Y)}{P(M(D_o+\mset{(x,\epsilon)}) = Y)}.
\end{align}
The ratio of the priors odds to the posterior odds quantifies the new  information gained by the adversary by observing the mechanism output.
A more private mechanism should have odds ratio close to one.
By the AHDP guarantee, we have $\left| \log \frac{P(M(D_o) = Y)}{P(M(D_o+\mset{(x,\epsilon)}) = Y)} \right| \leq \epsilon$.
Thus, we get the guarantee that 
\begin{equation}
    \frac{r_{\mathsf{posterior}}}{r_{\mathsf{prior}}} \in [ e^{-\epsilon}, e^{\epsilon}].
\end{equation}
In other words, in a Bayesian setup, AHDP ensures that posterior odds  does not deviate from the prior odds by more a factor that depends on the privacy demanded by the extra user. \\

\textit{Failure of Standard HDP Understood via the Proposed Hypothesis Testing Framework:} Another example that demonstrates the failure of standard HDP under correlations is the problem set up of \cref{claim:4}. Consider the standard HDP mechanism $M$ that outputs length of the dataset directly. Clearly, this mechanism satisfies the HDP definition (\cref{def:hetDP}) but $\cP(M,\cH(x,\epsilon)) = 1$ since the adversary can easily distinguish the hypotheses $\{D_o, D_o +   \mset{(x,\epsilon)} \}$ by the length of the dataset. Thus, if $\epsilon = 0$, then the adversary's power under AHDP is $\frac{1}{2}$, i.e., it's a random guess and the adversary can't infer any information about the dataset from the mechanism output but its power is $1$, i.e., it can infer the dataset exactly under standard HDP.

%% file: 4-algo.tex
\section{Universal AHDP Mechanisms} \label{sec:algo}

It should be noted that a $\cW$-AHDP certificate is specific to the correlation $\cW$.
A $\cW$-AHDP mechanism might not be $\cW'$-AHDP if $\cW' \nsubseteq \cW$.
For example, consider the setting presented in \cref{ex:ahdp} where we use a homogeneous $\epsilon_h$-DP mechanism $M_h$ to obtain a $\cW_h$-AHDP mechanism where $\cW_h = \cX \times \{\epsilon_h\}$.
Now suppose we have $\cW' = \cX \times \{\epsilon_h , \frac{\epsilon_h}{2} \}$, then $M_h$ is not $\cW'$-AHDP.

The domain of correlation, $\cW$, plays a key role in determining whether a mechanism is $\cW$-AHDP.
However, for real-world usage, modeling the domain $\cW$ can be a difficult task.
One possible way of mitigating this issue is via designing an interface where users can choose their privacy demand from a fixed list of privacy levels.
The users can choose the privacy level that is closest to their desired level.
Suppose the list of privacy levels is $L$ then we can set $\cW = \cX \times L$.
While this is a valid strategy, we focus on a more ambitious question  -- can you design (non-trivial) $\cW$-AHDP mechanisms without knowing $\cW$?
Since these mechanisms are correlation agnostic, the domain needs to be $\cS(\cX \times \bbR_{\geq 0})$, i.e., they can handle arbitrary correlations in the dataset. 
For any possible underlying correlation $\cW \subset \cX \times \bbR_{\geq 0}$ with finite $|\cW|$, the mechanism must be $\cW$-AHDP.
We call these mechanisms universal AHDP and define it in \cref{def:uAHDP}.

\begin{definition} \label{def:uAHDP}
    A mechanism $M:\cS(\cX \times \bbR_{\geq 0}) \to \cY$ is said to be universal AHDP if it is $\cW$-AHDP for all  $\cW \subset \cX \times \bbR_{\geq 0}$ with finite $|\cW|$. 
\end{definition}

Therefore, universal AHDP mechanisms, can be used in applications where the correlations are not known apriori.
Surprisingly, we show it is indeed possible to construct simple non-trivial universal AHDP mechanisms, making them attractive for real-world use.
However, before we describe such mechanisms, we highlight a shortcoming of universal AHDP mechanisms -- it is impossible for such mechanisms to be unbiased estimators for meaningful statistics of the dataset.
Let $f:\cX \to \bbR$ represent a generic function of interest and suppose we want to estimate the empirical statistic $f(D) = \sum_{(x,\epsilon)} f(x)h_D(x,\epsilon)$.
  \cref{prop:Bias} below shows that it is, in general, impossible to form  an unbiased AHDP-private estimate of $f(D)$ without any assumptions on the structure of $\cW$ or $f$.

\begin{restatable}{proposition}{propBias} \label{prop:Bias}
	In general, there exists no universal AHDP mechanism $M$ that can satisfy $\bbE[M(D)] = f(D)$.
\end{restatable}
\begin{proof}
	Consider the special case where $f(x) = 1 \forall\  x \in \cX$. In other words, we are interested in  estimating the size of the dataset.
    Suppose $\cW$ is such that $\exists x_o \in \cX$, $(x_o,0) \in \cW$ then by AHDP property, the mechanism's output distributions for datasets $D$ and dataset $D' = D + \mset{(x_o,0)}$ are equal, which makes it impossible for both $\bbE[M(D)] = |D|$ and $\bbE[M(D')] = |D'|$ to hold simultaneously.
\end{proof}

A short note is presented in \cref{rem:bias} on bounding the bias in terms of covariance of the data and privacy demand.
Although \cref{prop:Bias} is a negative result, it requires a worst-case setup.
It is still meaningful to ask how universal AHDP mechanisms can perform on real-world datasets.
Thus, shall focus on universal AHDP mechanisms for the rest of this section.
First, we refer to the sample mechanism suggested by \citet{Jorg15} that satisfies universal AHDP.
We then consider the task of mean estimation and linear regression, and further provide an alternate universal AHDP mechanism.
Both of the mechanisms are extremely simple to implement and deploy in real-world applications.

\subsection{Sample Mechanism of \citet{Jorg15}}
\label{sec:Jorg}

\citet{Jorg15} studied the problem of personalized privacy, without the complication of correlations between data and privacy demand.
However, their proposed Sample Mechanism satisfies the property of being universal AHDP.
The sample mechanism works in two stages.
In the first stage, the data is sub-sampled. 
A threshold $t > 0$ and a privacy mapping $\alpha(x,\epsilon)$ is chosen beforehand. 
All datapoints having $\alpha(x,\epsilon) \geq t$ are included in the sub-sampled dataset while rest of the data points are included in the dataset with diminishing probability as the privacy demand increases.
In particular, let $T(D) = \{(x,\epsilon) : h_D(x,\epsilon) > 0\}$.
From the original dataset $D$, obtain a randomly sub-sampled dataset $D'$ where
\begin{equation}
	h_{D'}(x,\epsilon) \sim \begin{cases}
		h_D(x,\epsilon) & \text{ if } \alpha(x,\epsilon) \geq t \\
		\mathsf{Bin}\left(h_D(x,\epsilon), \frac{e^{\alpha(x,\epsilon)}-1}{e^t - 1} \right) & \text{ else,}
	\end{cases} 
\end{equation}
$\forall (x,\epsilon) \in T(D)$. 
In other words, each datapoint is sampled independently with probability $ \frac{e^{\alpha(x,\epsilon) \wedge t}-1}{e^t - 1}$, where $a \wedge b$ denoted $\min\{a,b\}$.
In the second stage, any homogeneous $t$-DP (add-remove model) can be employed on the sub-sampled dataset $D'$ for the desired task.
The mechanism is $(\alpha(x,\epsilon) \wedge t)$-AHDP; a complete proof of privacy for the sample mechanism is presented in \cref{sec:SM}.
Thus, setting $\alpha(x,\epsilon) \leq \epsilon$ ensures universal AHDP.

The parameter $t$ controls the bias and variance of the mechanism in a certain sense -- if $t$ is set high then much of the data is absent in the sub-sampled dataset $D'$ introducing a possible bias in the output of the mechanism, but the noise due to $t$-DP mechanism can be expected to be low.
Conversely, a lower value of $t$ ensures the dataset $D'$ is more representative of the true data but the noise due to the $t$-DP mechanism will be higher in general.
It must be noted that the sample mechanism is remarkably versatile and can be used for any task by employing the corresponding homogeneous DP mechanism in the second step.

\citet{Jorg15} also propose an adaptation of the Exponential Mechanism \cite{McSherry07} that satisfies universal AHDP. 
We do not discuss this mechanism since our focus is on simple and efficient mechanisms.

\subsection{Additional Mechanisms Based on Linearity} \label{sec:LQ}

We provide some alternative universal AHDP mechanisms for task of linear counting queries, which encompass sum, count, and histogram estimation.
The ideas involve readily extended to mean estimation and empirical risk minimization.
The key idea behind these mechanisms  is to have a carefully linear combination of the data with some additional Laplace noise.

Recall that $\cL(\lambda)$ denotes a sample from the Laplace distribution with pdf given by $p(x;\lambda) = \frac{1}{2\lambda} e^{-|x|/\lambda}$.

Given a $f:\cX \to \bbR$, suppose  we are interested in estimating the linear counting query $f(D) = \sum_{x \in \cX} f(x)h_D(x)$.
Using the linearity, we can estimate $f(D)$ using a weighted sum over all the datapoints in $D$.
In \cref{prop:newDP}, we present a class of mechanisms for such linear query estimation that are universal AHDP. 
Let $l = \min_{x \in \cX} f(x)$ and $h = \max_{x \in \cX}f(x)$.

\begin{proposition}[AHDP Linear Query] \label{prop:newDP}
    The mechanism $$M(D) = l + \sum_{(x,\epsilon) \in T(D)}\alpha(x,\epsilon)(f(x)-l)h_D(x,\epsilon) + \cL(h-l)$$ is  satisfies $\alpha$-AHDP. Thus, for choices of $\alpha(x,\epsilon) \leq \epsilon$, the mechanism is universal AHDP.
\end{proposition}
\begin{proof}
    Suppose the correlation is given by unknown $\cW$. 
    For any pair of datasets $D,D' \in \cS(\cW)$, it suffices to consider the ratio of pdfs of the mechanism.
    Let $T(D) = \{(x,\epsilon) : h_D(x,\epsilon) > 0\}$ and note that $T(D), T(D') \subseteq \cW$.
    Thus, $\forall s \in \bbR$, 
    \begin{align*}
        &\left| \log \frac{p(M(D) = s)}{p(M(D') = s)} \right|  \\
        &=\bigg|\frac{-|s - l - \sum_{(x,\epsilon)\in T(D)}\alpha(x,\epsilon)(f(x)-l)h_D(x,\epsilon)|}{h-l} \nonumber \\
         &~~+ \frac{|s - l - \sum_{(x,\epsilon) \in T(D')}\alpha(x,\epsilon)(f(x)-l)h_{D'}(x,\epsilon)|}{h-l} \bigg| \\
        &\leq \left|\frac{\sum_{(x,\epsilon) \in T(D) \cup T(D')}\alpha(x,\epsilon)(f(x)-l)(h_{D'}(x,\epsilon)-h_D(x,\epsilon))}{h-l}\right| \\
        &\leq \frac{\sum_{(x,\epsilon) \in T(D) \cup T(D')}\alpha(x,\epsilon)\left|f(x)-l\right| \left|h_{D'}(x,\epsilon)-h_D(x,\epsilon)\right|}{h-l} \\
        &\leq \sum_{(x,\epsilon) \in T(D) \cup T(D')}\alpha(x,\epsilon)\left|h_{D'}(x,\epsilon)-h_D(x,\epsilon)\right| \\
        &= \sum_{(x,\epsilon) \in \cW}\alpha(x,\epsilon)|h_{D'}(x,\epsilon) - h_{D}(x,\epsilon)|.
    \end{align*}
    Thus, choosing  $\alpha(x,\epsilon) \leq \epsilon$ ensures $\cW$-AHDP for any $\cW$, proving that the mechanism is universal AHDP.
\end{proof}

\cref{prop:newDP} shows that an affine estimator with weights $\alpha(x,\epsilon) \leq \epsilon$ will satisfy universal AHDP.
A natural choice is to consider $\alpha(x,\epsilon) = \epsilon$ but it may be noted that a single user having no privacy demand $(\epsilon \to \infty)$ would skew the estimate arbitrarily.
Thus, choosing weights $\alpha(x,\epsilon) = 1 - e^{-\epsilon}$ or $\alpha(x,\epsilon) = \frac{\epsilon}{1+\epsilon}$, may be more appropriate depending on the context.

Mechanisms for estimate the sum, counting the size of the dataset immediately follow from \cref{prop:newDP} using $f(x) = x$ (assuming $\cX$ bounded) and $f(x) = 1$, respectively.
Mechanisms for histogram over $k$ bins ($\cX = [k]$) follow similarly.

\begin{example}[Sum Estimation] \label{ex:sum}
For illustration, we provide an example of an universal AHDP mechanism for sum estimation. Let $\cW = \{ (0,0), (1,1), (2,\infty)\}$ and let the dataset $D$ be such that $h_D(0,0) = 10$, $h_D(1,1) = 10$, and $h_D(2,\infty) = 10$, i.e., there are $10$ samples for each pair of data-privecy tuple in the domain.
For estimators of form given in \cref{prop:newDP}, $l=0$, and $h=2$.
Choosing $\alpha(x,\epsilon) = \epsilon$, the mechanism in \cref{prop:newDP} would always output $\infty$ -- demonstrating the issue with setting $\alpha(x,\epsilon) = \epsilon$.
Choosing $\alpha(x,\epsilon) = 1 - e^{-\epsilon}$, the mechanism output is $0 + (1-e^{-1}) \times 10 + (1-e^{-2}) \times 10 \times 2 + \cL(2) = 23.6 + \cL(2)$.
Thus, the output has a bias of $-6.4$.
\end{example}

\paragraph{\textbf{Mean estimation:}}
Mean estimation plays a key role in statistics.
In the traditional setting, mean estimation refers to population mean estimation since sample mean can be obtained directly.
In the private setting, we need to estimate the sample mean, which is then a good estimate for the population mean for large sample sizes.
The sample mean can be estimated via dividing the sum estimate by the count estimate.
Using the composition property presented in \cref{prop:BC}, we can bound the privacy of the mechanism.
 \cref{prop:newDP-mean} makes this idea concretely for $\cX \subseteq [l,h]$.
 Similar idea can be extended to get universal AHDP mechanism for frequency estimation and is presented with proof in \cref{sec:exp-det}.

\begin{proposition}[AHDP Mean] \label{prop:newDP-mean}
	The mechanism 
	$$M(D) = \frac{l + \sum_{(x,\epsilon) \in T(D)}\alpha_1(x,\epsilon)h_D(x,\epsilon)(x-l) + \cL((h-l))}{\sum_{(x,\epsilon) \in T(D)}\alpha_2(x,\epsilon)h_D(x,\epsilon) + \cL(1)}$$ 
	satisfies $(\alpha_1 + \alpha_2)$-AHDP. Thus, for choices of $\alpha_1(x,\epsilon) + \alpha_2(x,\epsilon) \leq \epsilon$, the mechanism is universal AHDP.
\end{proposition}

Consider $\alpha_1 = \alpha_2 = \alpha$.
Choosing $\alpha(x,\epsilon)$ as large as possible, i.e., $\alpha(x,\epsilon) = \epsilon/2$ ensures that the effect of Laplace noise in both the numerator and denominator is minimal; however, this comes at the cost of skewing the estimate towards user data who desire less privacy.
For example, if there is a user demanding no privacy $(\epsilon \to \infty)$, then the data of all users demanding some privacy gets ignored. 
To reduce the skew, one would want $\alpha(x,\epsilon)$ to vary as little as possible, such as  $\alpha(x,\epsilon) = \frac{1 - e^{-\epsilon}}{2}$  but this tendency would increase the effect of the Laplace noise. 
In \cref{sec:exp}, we consider three choices of $\alpha(x,\epsilon)$ experimentally.

\begin{example}[Mean Estimation] \label{ex:mean}
For illustration, consider the same setup as that of \cref{ex:sum} for mean estimation. 
That is, $\cW = \{ (0,0), (1,1), (2,\infty)\}$ and the dataset $D$ is such that $h_D(0,0) = 10$, $h_D(1,1) = 10$, and $h_D(2,\infty) = 10$.
For estimators of form given in \cref{prop:newDP-mean}, choosing $\alpha(x,\epsilon) = \epsilon/2$ would always output the mean as $2$ if we consider the privacy demand of $\infty$ as a limit.
Choosing $\alpha(x,\epsilon) = (1 - e^{-\epsilon})/2$, the mechanism output is $\frac{11.8 + \cL(2)}{7.48 + \cL(1)}$.
For numerical stability, we set the denominator as $\max\{7.48 + \cL(1), 1\}$, which does not affect the privacy guarantees by \cref{prop:pp}. We present the empirical distribution of the output over $10^7$ samples in \cref{fig:emp-dist}. The mean of the distribution is about 1.6.
 
\begin{figure}
    \centering
    \includegraphics[width=0.75\linewidth]{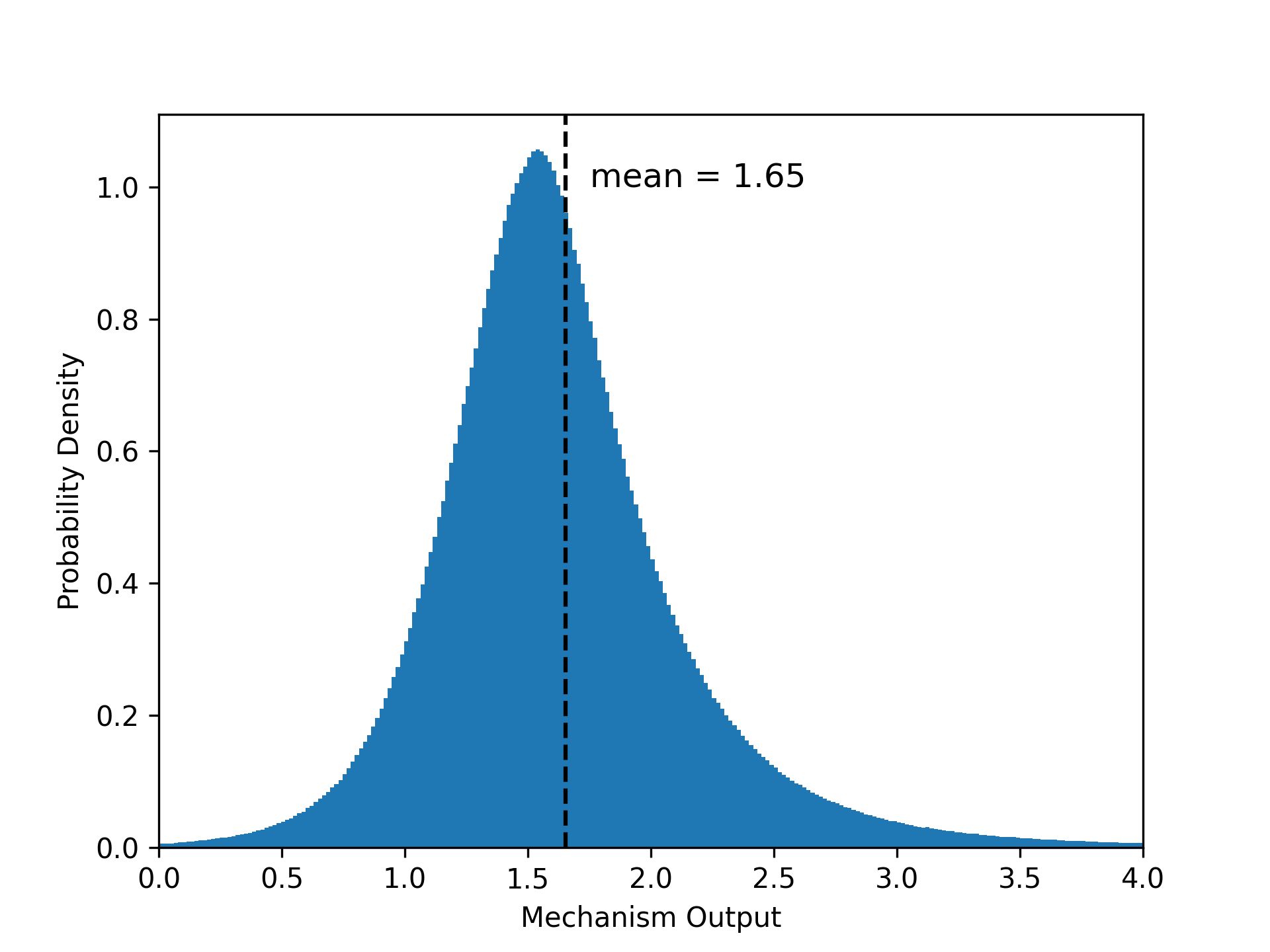}
    \caption{Empirical distribution of the output of the mechanism on \cref{ex:mean}.}
    \label{fig:emp-dist}
\end{figure}
\end{example}

\begin{remark}[Asymptotic Bias for Mean Estimation] \label{rem:bias}
We comment on the asymptotic bias for mean estimation here.
    Suppose $\{(X_i,\epsilon_i)\}_{i=1}^n$ are generated i.i.d. from joint distribution $P_{X,\epsilon}$ and we are interested in finding the mean of $P_X$ as the number of samples increases.
For simplicity, assume $\cX = [0,1]$. 
Using the mechanism in \cref{prop:newDP-mean}, we have
\begin{equation}
    M(D_n) = \frac{\sum_i \alpha(x_i,\epsilon_i) x_i + \cL(2) }{\sum_i \alpha(x_i,\epsilon_i) + \cL(2)}.
\end{equation}
As $n$ increases, we have $\frac{\sum_i \alpha(x_i,\epsilon_i) x_i + \cL(2)}{n}   \to \bbE[X\alpha(X,\epsilon)]$ almost surely.
Similarly, $\frac{\sum_i \alpha(x_i,\epsilon_i) + \cL(2)}{n}   \to \bbE[\alpha(X,\epsilon)]$ almost surely by strong law of large numbers.
Assuming $\bbE[\alpha(X,\epsilon)] > 0$, we have $\lim_{n \to \infty} M(D_n) \to \frac{\bbE[X\alpha(X,\epsilon)]}{\bbE[\alpha(X,\epsilon)]}$ almost surely using the continuous mapping theorem.
Rewriting, we have $\lim_{n \to \infty} M(D_n) \to \bbE[X] +  \frac{\mathsf{Cov}(X,\alpha(X,\epsilon))}{\bbE[\alpha(X,\epsilon)]}$ almost surely.
Assuming $\alpha(x,\epsilon)$ is only dependent on $\epsilon$, to have low bias we need low covariance between data and privacy demand and high mean value for $\alpha(x,\epsilon)$.
The asymptotic bias is $0$ in case $X$ and $\epsilon$ are independent.
\end{remark}

\paragraph{\textbf{Linear Regression:}} 
We show that one can design universal AHDP mechanisms for linear regression based on similar ideas as that of \cref{prop:newDP}.
The method can be extended to general empirical risk minimization.

In the non-private regime, given a dataset $D = \{(x_i,y_i)\}_{i=1}^n$ drawn from some distribution i.i.d., the goal is to find 
\begin{equation}
    h^* = \arg\min_{h \in \cH} \frac{1}{n} \sum_{i=1}^n l(h(x_i),y_i),
\end{equation}
where $\cH$ is the function class considered and $l$ is a loss function.
The expected sub-optimality on the obtained $h^*$ can then be bound using statistical learning theory based on the function class $\cH$.
In the private setting, each input tuple $(x_i,y_i)$ is modeled as an user data and the output $h^*$ needs to be private with respect to each user's data.

Concretely, from here on, we shall denote the domain of user data by $\cZ = \cX \times \cY$ instead of the prevailing notation $\cX$ to separate the covariates and the target. 
Focusing on linear regression, assume $\cX = [-1,1]^d$ and $\cY = [-1,1]$, i.e., the covariates and the targets are in the $l_{\infty}$ ball of radius $1$ around the origin.
The hypothesis class $\cH$ is the set of all linear maps, i.e., we write $h_{\theta}(x) = \theta^Tx$ and $\theta \in \bbR^{d}$. 
The loss function is given by $l(\theta^Tx,y) = |\theta^Tx - y|^2$.
Each user also specifies a privacy demand so each user's input is the tuple $(z,\epsilon)$, where $z = (x,y)$.
Let $\cW$ denote the possible correlations, i.e., $(z,\epsilon) \in \cW \subseteq \cZ \times \bbR_{\geq 0}$. 

Let $M$ be a mechanism that outputs a regression coefficient $\hat{\theta}$. 
For the $\cW$-AHDP requirement, we need $\left|\frac{\Pr\{M(D) \in S\} }{\Pr\{M(D') \in S\}} \right| \leq d_{\alpha}(D,D')$ and $\alpha(z,\epsilon) \leq \epsilon$.
We present a mechanism based on the Functional Mechanism presented in \citet{Zhang12functional} in Mechanism~\ref{alg:LR}.
The mechanism arbitrarily orders the data and constructs the standard covariate matrix and target vector. 
The ordering does not matter in the operations considered.
$\mathsf{Diag}(\vecw)$ denotes to the diagonal matrix constructed with the vector $\vecw$.

In a nutshell, instead of optimizing $\sum_i |x_i^T \theta - y_i|^2$, we aim to optimize $\sum_{i} \alpha(z_i,\epsilon_i) |x_i^T \theta - y_i|^2$, which allows us to down-weigh the users demanding high privacy.
In particular, 
\begin{equation}
    \sum_{i} \alpha(z_i,\epsilon_i) |x_i^T \theta - y_i|^2 = (X\bm\theta - \vecy)^T\mathsf{Diag}(\vecw) (X\bm\theta - \vecy).
\end{equation}
The optimal regression coefficient can be calculated in a straightforward manner to be $\bm\theta^* = (X^T \mathsf{Diag}(\vecw) X)^{-1} X^T \mathsf{Diag}(\vecw) \vecy$.
However, we need to add some Laplace noise to ensure privacy.
The proof of privacy for Mechanism~\ref{alg:LR} can be found in \cref{sec:LRP}.
In \cref{sec:exp}, we present the performance of Mechanism~\ref{alg:LR} for three choices of $\alpha(z,\epsilon)$.

\floatname{algorithm}{Mechanism}
\begin{algorithm}
   \caption{$\alpha$-AHDP Linear Regression}
   \label{alg:LR}
\begin{algorithmic}
   \STATE {\bfseries Input:} $\alpha: \cW \to \bbR_{\geq 0}$, $D = \mset{(z_i,\epsilon_i)}_{i=1}^n$ with $z_i = (x_i,y_i)$.
   \vspace{2pt} \STATE Arbitrarily order the dataset and construct
   \STATE $X \in [-1,1]^{n \times d}$, $\vecy \in [-1,1]^{n \times 1}$, $\vecw \in \bbR_{\geq 0}^{n \times 1}$ with $i$-th row as $x_i$, $y_i$, $\alpha(z_i,\epsilon_i)$ respectively.
   \vspace{2pt} \STATE  Sample $N_A \in \bbR^{d \times d}$, $N_A(i,j) \sim \cL(d^2 + d) \ \forall i,j$
   \vspace{2pt} \STATE Sample $N_b \in \bbR^{d \times 1}$, $N_b(i) \sim \cL(d^2 + d) \ \forall i$
   \vspace{2pt}\STATE $A \gets X^T \mathsf{Diag}(\vecw) X + N_A$
   \vspace{2pt}\STATE $b \gets X^T \mathsf{Diag}(\vecw) \vecy + N_b$
   \vspace{2pt}\STATE Return $A^{-1}b$
\end{algorithmic}
\end{algorithm}

%% file: 6-exp.tex
\section{Experiments} \label{sec:exp}

We consider three experiments in this section -  mean estimation, relative frequency estimation, and linear regression.
These are arguably the three most fundamental tasks in the real-world.
The purpose of these experiments is to 
\begin{enumerate}
      \item understand the behavior of the universal AHDP mechanisms presented in \cref{sec:algo}, and 
    \item evaluate how much realistic datasets deviate from the worst-case perspective of \cref{prop:Bias}.
\end{enumerate}

For all the experiments, we consider three linear-query based mechanisms as described in \cref{sec:LQ} with
\begin{enumerate}
    \item $\alpha(x,\epsilon) \propto \epsilon$,
    \item $\alpha(x,\epsilon) \propto 1-e^{-\epsilon}$, and
    \item $\alpha(x,\epsilon) \propto \frac{\epsilon}{1+\epsilon}$.
\end{enumerate}
For mean and frequency estimation, since we use composition, we divide the $\alpha$ budget in half for the numerator and denominator.
For linear regression, there is no such consideration.
In addition to these three mechanisms, we also consider the Sampling Mechanism of \cref{sec:Jorg} using $\alpha(x,\epsilon) = \epsilon$ and two different threshold values.
Due to lack of a publicly available dataset on this topic, we synthetically generate the datasets from a realistic distribution with the help of a Large Language Model (LLM); we used OpenAI's GPT-4o \cite{hurst2024gpt} for the experiments. \\

In the first experiment, we query the LLM to generate tuples of weight (in kg) and privacy demand several times independently (see \cref{sec:exp-det} for details).
The LLM generated dataset naturally exhibits a high degree of correlation between the weights and the privacy demand, with a correlation coefficient of $ -0.84$.
The dataset is visualized in \cref{fig:corr}.
We set $l = 0$ and $h=150$.
The LLM was instructed to keep the privacy demand $\epsilon \in [0,3]$ so we try sampling mechanism with $t=2$ and $t=0.5$.
The dataset consists of about $2200$ samples.
For a given size $n$, we randomly sub-sample $n$ points from the dataset and run the five mechanisms on the sub-sampled dataset over $5000$ trials.
The mean-squared error with respect to the true mean of the sub-sampled dataset is plotted as $n$ is increased in \cref{fig:MSE}.
The Sampling Mechanism with $t=0.5$ outperforms other methods with vanishing error at large sample sizes while the other methods seem to incur a constant error even at large sample sizes.
The linear query-based mechanism with $\alpha(x,\epsilon) = \frac{\epsilon}{2}$ performs worse than the other two linear query-based alternatives. \\

In the second experiment, we query the LLM to generate tuples of education level and privacy demand, where the privacy demand is restricted to a list of five levels (see \cref{sec:exp-det} for details).
The dataset consists of $3000$ such tuples and the chi-squared test for independence yields a p-value of $0$, i.e., below the floating point precision.
A low p-value in the chi-squared test of a contingency table indicates that it is quite unlikely that privacy demand and education levels are independent \cite{mchugh2013chi}.
The dataset is visualized in \cref{fig:hist}.
We obtain the relative frequencies for the education level from the above-mentioned three linear query-based methods and the sampling mechanism with $t=0.1$ and $t=1$. 
The mechanism is obtained by independently estimating each education level by dividing the corresponding count estimate with the total number estimate; see \cref{sec:exp-det} for details.
For a given sample size $n$, we randomly sub-sample a size $n$ dataset and estimate the relative frequencies of the six categories and we record the corresponding $l_{\infty}$ for the five methods. 
Over 5000 trials, we plot the empirical mean $l_{\infty}$ error in \cref{fig:li}.
The Sampling Mechanism with $t=0.1$ outperforms the other methods but does not attain vanishing error unlike the mean estimation experiment; the other methods incur significant error. \\

For the third experiment, we use the California Housing dataset \cite{pace1997sparse} to obtain the covariates and target value for regression. 
There are eight covariates in the dataset.
We split the dataset into 18000 training samples and about 2000 test samples.
The size of the dataset makes it infeasible to query an LLM to get the corresponding privacy demand.
We sample privacy demand i.i.d. according to $\log \epsilon \sim \mathsf{Uniform}([-5,2])$.
Although the dataset generated using this process does not explicitly contain correlations between data and privacy, it is still useful in comparing the performance of the different universal AHDP mechanisms.
We normalize the dataset to satisfy the assumptions of linear regression presented in \cref{sec:algo}, i.e., we linearly rescale the covariates and targets to be in $[-1,1]$.
Mechanism~\ref{alg:LR} is used with the three proposed $\alpha$ above and for the sampling mechanism, the functional mechanism of \citet{Zhang12functional} is used for regression on the sub-sampled dataset.
\cref{fig:LR} shows the median residual test error as the dataset size is increases for the methods along with the non-private least squares solution.
For each sample size, we randomly sampled a subset of the training set and performed 200 trials.
Increasing the training set size generally leads to better performance as it compensates for the noise injected due to privacy constraints.
It remains unclear whether we can expect the error to reach the non-private level under a more realistic correlation between data and privacy.

\begin{figure}
    \centering
    \includegraphics[width=0.45\textwidth]{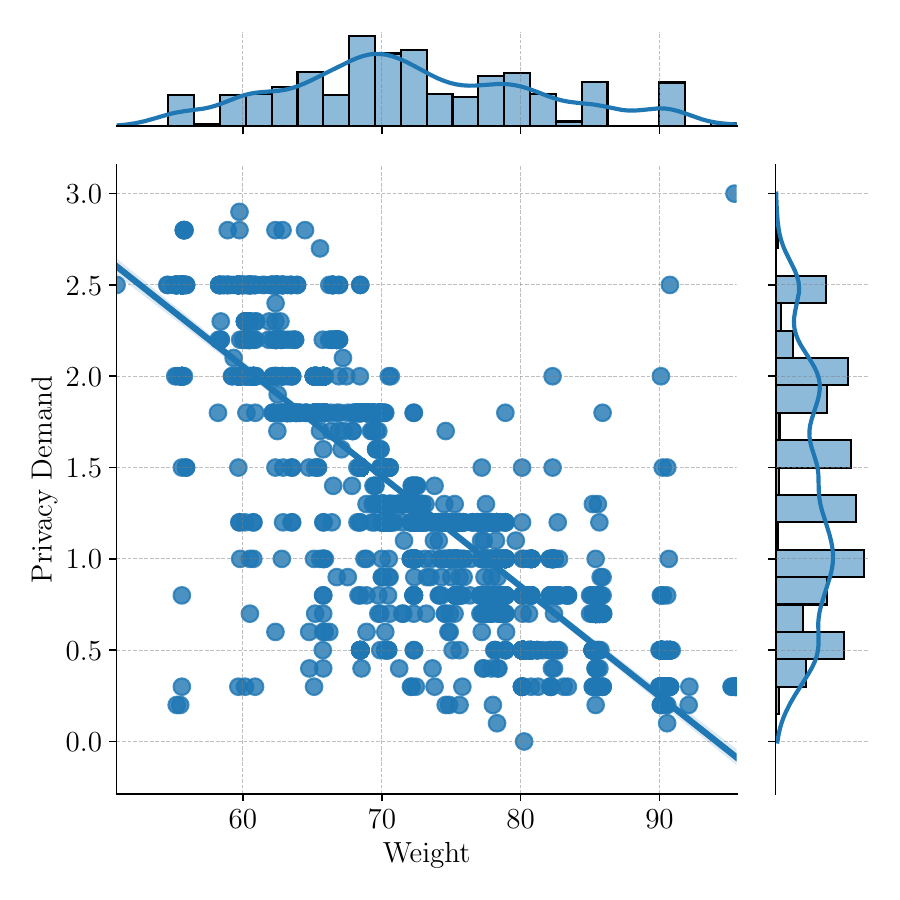}
    \caption{Scatter plot of LLM-generated dataset of weight (in kg) and  privacy demand ($\epsilon$, higher is less privacy). The best fit regression line shows a negative correlation between weight and $\epsilon$.}
    \Description{A scatter plot of weight (in kg) and privacy demand $(\epsilon)$.}
    \label{fig:corr}
\end{figure}

\begin{figure}
    \centering
    \includegraphics[width=0.5\textwidth]{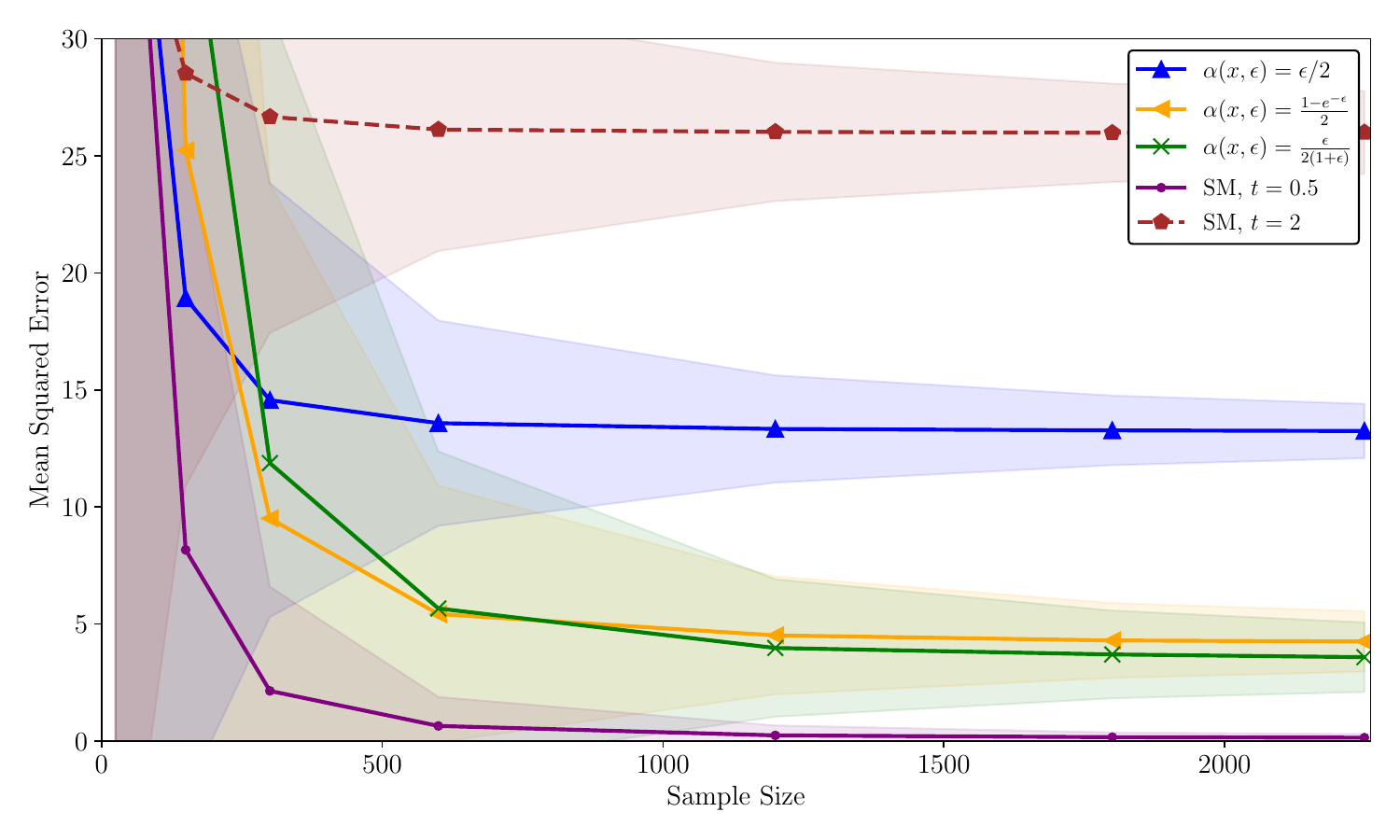}
    \caption{MSE vs number of samples for different methods for the task of mean estimation.}
    \Description{A plot showing decaying error for the five methods as number of samples is increased.}
    \label{fig:MSE}
\end{figure}

\begin{figure}
    \centering
    \includegraphics[width=0.45\textwidth]{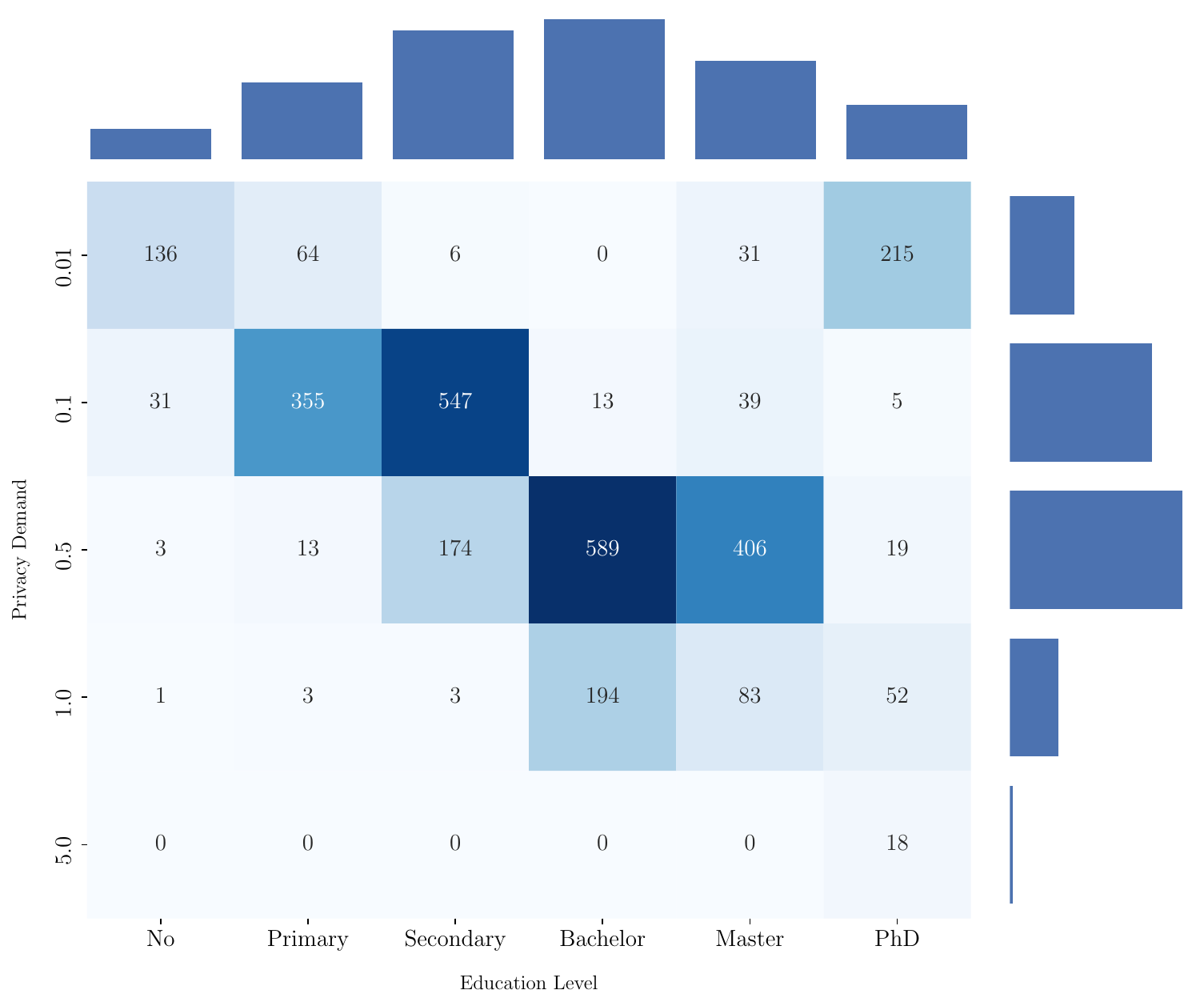}
    \caption{Contingency table of LLM-generated dataset of education level and  privacy demand ($\epsilon$, higher is less privacy).}
    \Description{The education levels are 'no', 'primary', 'secondary', 'bachelor', 'master', and 'Ph.D.'. The privacy levels are $0.01$, $0.1$, $0.5$, $1$, and $5$. }
    \label{fig:hist}
\end{figure}

\begin{figure}
    \centering
    \includegraphics[width=0.5\textwidth]{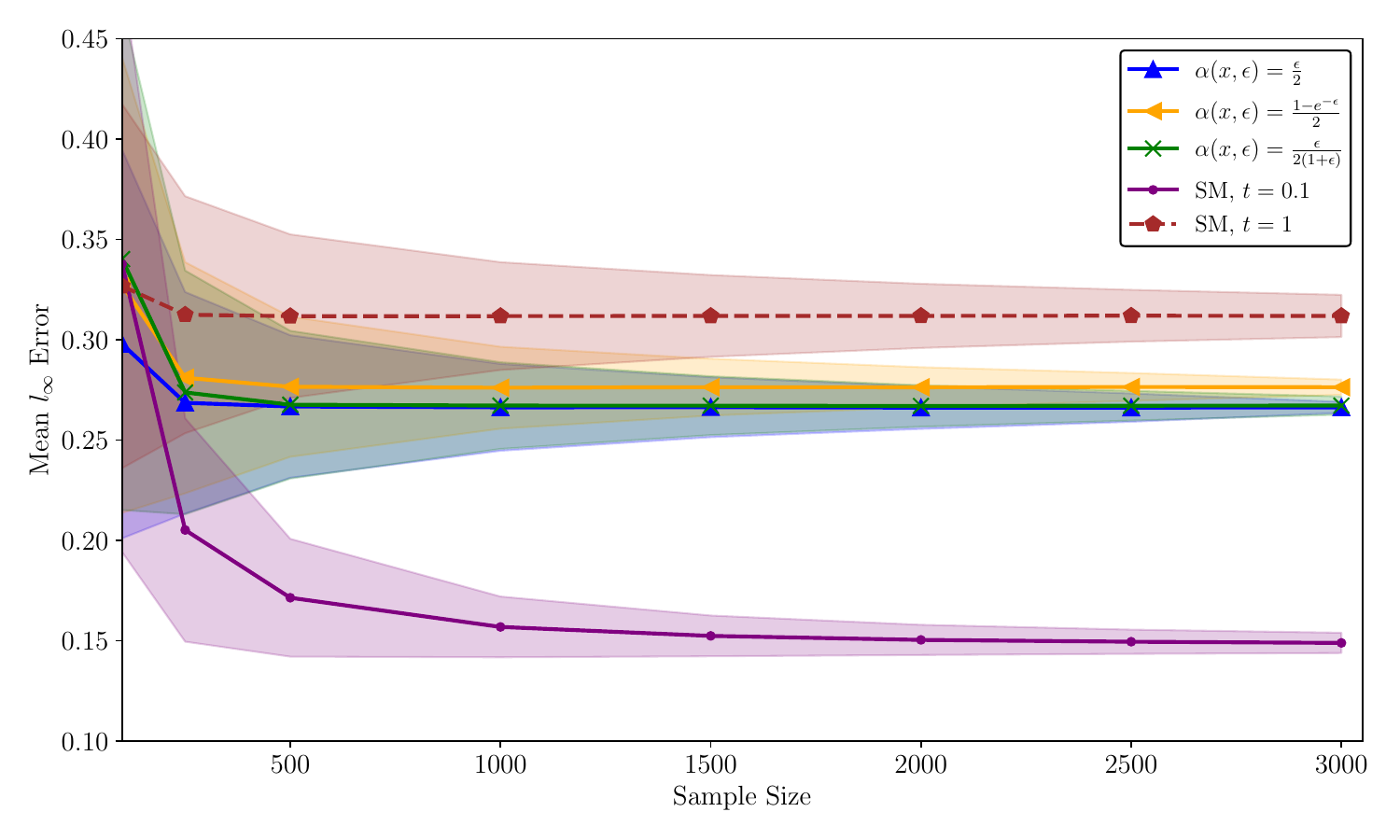}
    \caption{Mean $l_{\infty}$ error of the five methods as the sample size is increased for histogram estimation.}
    \Description{SM with $t=0.1$ has lower $l_{\infty}$ error than the other methods.}
    \label{fig:li}
\end{figure}

\begin{figure}
    \centering
    \includegraphics[width=0.5\textwidth]{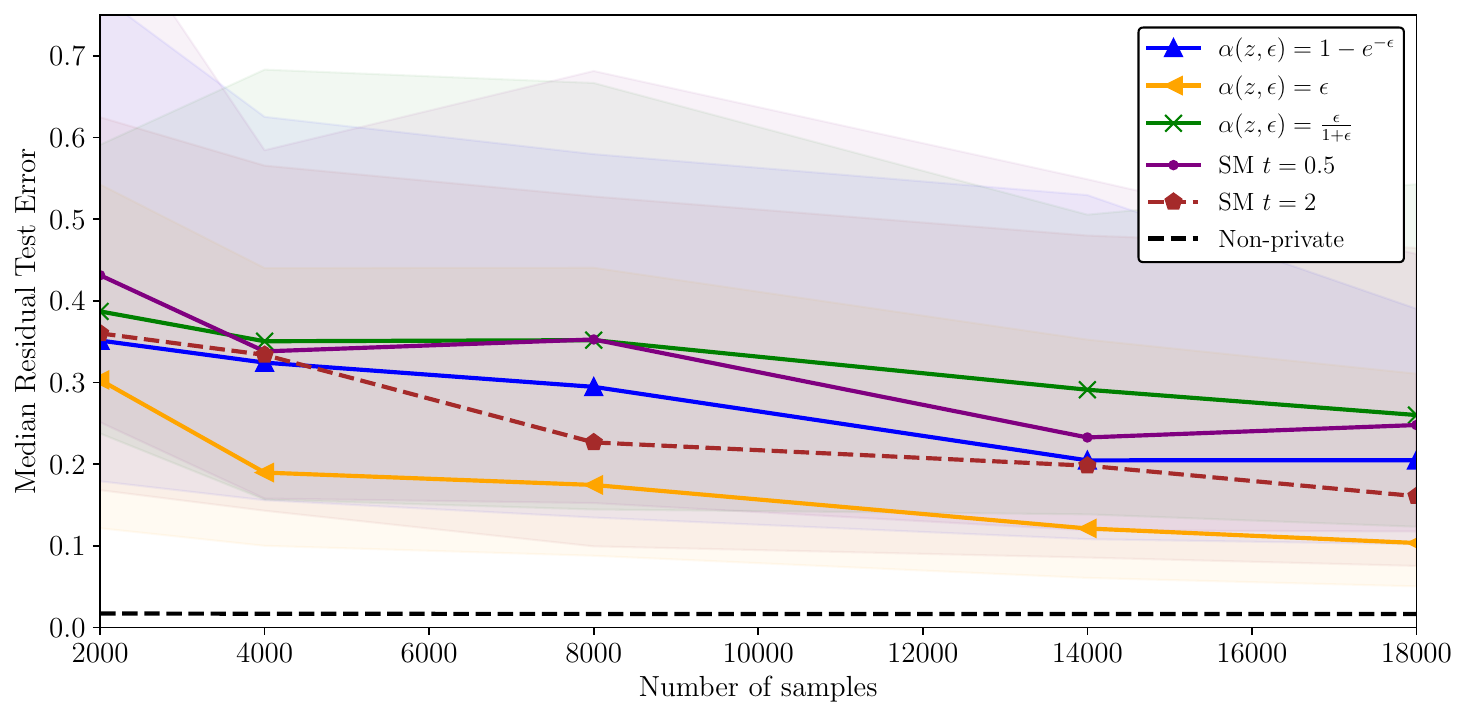}
    \caption{Median residual test error for the five methods as number of training samples is increased in the task of linear regression.}
    \Description{Linear Query based method with $\alpha(x,\epsilon) = \epsilon/2$ has lower error than the other methods.}
    \label{fig:LR}
\end{figure}

%% file: 7-con.tex
\section{Conclusion and Future Work} \label{sec:con}

This work investigates Heterogeneous Differential Privacy (HDP) under potential correlations between user data and user privacy demand—an aspect largely overlooked in the literature. 
We demonstrate that standard HDP definitions, when used under such correlations, can result in mechanisms that fail to provide meaningful privacy guarantees.
We address these issues by proposing Add-remove Heterogeneous Differential Privacy (AHDP)—a privacy notion based on an add-remove model of neighboring datasets that jointly considers the domain of user data and user privacy demand.

We present a general hypothesis testing framework to analyze the privacy guarantees of different privacy definitions. 
We provide privacy guarantees for AHDP, along with guarantees for standard homogeneous DP, in this hypothesis testing framework.
We further demonstrate how universal AHDP mechanisms—those that do not assume knowledge of data-privacy correlations—can be constructed for core tasks such as mean estimation, frequency estimation, and linear regression. These mechanisms, being correlation-agnostic, are practically attractive and easy to implement.

While universal AHDP mechanisms are robust to privacy-data correlations, they can suffer from poor performance in the worst-case due to inherent bias in the estimates.
Such trade-offs may motivate a broader philosophical debate on whether it is better to personalize privacy  or offer a uniform privacy level to all users.
However, it is important to note that even standard homogeneous DP may suffer from implicit selection bias -- more privacy-demanding users may not engage with the service provider.
Thus, in this light, homogeneous DP with add-remove neighbors can be seen as a universal-AHDP mechanism that provides $\epsilon$-privacy only to participating users and nothing to those who self-exclude due to privacy concerns.

There are several avenues of future work, such as relaxing AHDP similar to $(\epsilon,\delta)$-DP, considering tasks like statistical learning under AHDP, and so on.
Another line of future work could be to empirically evaluating how conventional HDP mechanisms fail under correlations—particularly via membership inference attacks on real-world datasets.

%% file: A-appendix.tex
\section{Multisets: Notations} \label{A:multisets}

A multiset $D$ containing elements from $\cX$ is represented as the tuple $D = (\cX,h)$ where $h: \cX \to \bbZ_{\geq 0}$ is the count of each element in $D$.

We occasionally use the denote a multiset using the notation $\{\{ \cdot \}\}$.
For example, for $a,b,c \in \cX$, the notation $\{\{a,a,b,c\}\}$ refers to a multiset with $h(a) = 2,\ h(b) = 1,\ h(c) = 1$.

\paragraph{Addition and Subtraction:} for $D = (\cX,h)$, $D' = (\cX,h')$, $D + D'$ is the multiset $(\cX,h + h')$. Similarly, for subtraction we have $D - D' = (\cX, (h-h')_+)$, where $(h-h')_+ (x) = \max\{0,h(x) - h'(x)\}$.

\paragraph{Size:} for $D = (\cX,h)$, define the size of the multiset as $|D| = \sum_{x \in \cX} h(x)$.

\section{Proof of Hypothesis Testing Property} \label{A:HT}
\thmHT*
\begin{proof}
Type I error rate is given by $e_1(R,D,D') = \Pr{M(D) \in R}$.
 By AHDP, for any measurable $R \in \subseteq \cY$, we have $ \Pr\{M(D) \in R^c\} \leq e^{d_{\alpha}(D,D')} \Pr\{M(D') \in R^c\}$.
Thus,
\begin{align}
    \text{AHDP} &\Leftrightarrow \Pr\{M(D) \in R^c\} \leq e^{d_{\alpha}(D,D')} \Pr\{M(D') \in R^c\} \\ 
    &\hspace{2cm} \forall D,D' \in \cS(\cW), R \subseteq \cY, \\
    &\Leftrightarrow 1 - e_1(R,D,D') \leq e^{d_{\alpha}(D,D')} e_2(R,D,D')\\ 
    &\hspace{2cm} \forall D,D' \in \cS(\cW), R \subseteq \cY. 
\end{align}
It may be noted that \cref{thm:HT} also implies $  e^{d_{\alpha}(D,D')}e_1(R,D,D') + e_2(R,D,D') \geq 1 \ \forall R \subseteq \cY$ by noting that $e_2(R,D,D') = e_1(\cY \setminus R,D',D)$.
\end{proof}

\section{Proof of Composition} \label{A:comp}
\BasicCom*
\begin{proof}
    We shall prove the proposition for the base of $M(D) = (M_1(D),M_2(D,M_1(D)))$. 
    The general version follows analogously; one can also use induction.
    For any measurable $S \in \cY_1 \times \cY_2$, let $S_1 = \{y_1 | \exists y_2 \in \cY_2, (y_1,y_2) \in \cS\}$ and $S_2(x) = \{y_2 | (x,y_2) \in S\}$ then we have 
    \begin{align}
        \Pr\{M(D) \in S\} &=  \Pr\{(M_1(D),M_2(D,M_1(D)) \in S\}\\
        &\hspace{-2cm}= \int_{y_1 \in S_1} \Pr\{M_2(D,M_1(D)) \in S_2(y_1) | M_1(D) = y_1 \} dP_{M_1(D)}(y_1) \\
        &\hspace{-2cm}= \int_{y_1 \in S_1} \Pr\{M_2(D,y_1) \in S_2(y_1)\} dP_{M_1(D)}(y_1) \\
        &\hspace{-2cm}\leq \int_{y_1 \in S_1} e^{d_{\alpha_2}(D,D')}\Pr\{M_2(D',y_1) \in S_2(y_1) \} dP_{M_1(D)}(y_1) \\
        &\hspace{-2.5cm}\leq \int_{y_1 \in S_1} e^{d_{\alpha_2}(D,D') + d_{\alpha_1}(D,D')}\Pr\{M_2(D',y_1) \in S_2(y_1) \} dP_{M_1(D')}(y_1) \\
        &\hspace{-2cm}= e^{d_{\alpha_2}(D,D') + d_{\alpha_1}(D,D')} \Pr\{M(D') \in S\}\\
        &\hspace{-2cm}= e^{d_{\alpha_2 + \alpha_1}(D,D')} \Pr\{M(D') \in S\}.
    \end{align}
\end{proof}

\section{Proof of Post-processing} \label{A:pp}
\PostProc*
\begin{proof}
For $S \in \cZ$, let $Y(S) = \{y \in \cY|K(y) \in S\}$.
Thus, we have
\begin{align}
    \Pr\{K(M(D)) \in S\} &= \Pr\{M(D) \in Y(S)\} \\
    &\leq e^{d_{\alpha}(D,D')} \Pr\{M(D') \in Y(S)\} \\
    &=  e^{d_{\alpha}(D,D')} \Pr\{K(M(D')) \in S\}.
\end{align}
\end{proof}

\section{Proofs Related to Power} \label{sec:cap}

\homoCap*

\begin{proof}
	(A) Without loss of generality, let $\cX = [k]$, $D_i = D_o \cup \{i\}$ for $i \in [k]$.
	Define $e_i = \Pr\{ \psi(M(D_i)) \neq i\}$, i.e., the probability of error under mechanism $M$ and adversary's classifier $\psi$ (which is based on knowledge of $M$).
	
	For a given mechanism $M$, and any $\psi$, by DP guarantee and post-processing, we have
	\begin{align}
		1-e_i \leq e^{\epsilon} \Pr\{\psi(M(D_j)) = i\}\ \forall i,j \in [k].
	\end{align}
	Adding over all $i \neq j$, we get
	\begin{equation}
		(k-1) \leq e^{\epsilon} e_j + \sum_{i \neq j} e_i.
	\end{equation}
	Letting $\pi_o(x) = \frac{1 + \bm1 \{x = j \}(e^{\epsilon} - 1) }{e^{\epsilon} + k - 1}$, we get
	\begin{align}
		\cP(M,\cH) &\leq \max_{\psi} \sum_{i = 1}^k \pi_o(i) (1- e_i) \label{eq:cvx} \\ 
		&\leq 1 - \frac{k-1}{e^{\epsilon} + k - 1} \\
		&= \frac{1}{1 + (k-1)e^{-\epsilon}},
	\end{align}
    where \eqref{eq:cvx} follows since $\min_{i} (1 - e_i) \leq \sum_{i} \pi(i)(1-e_i)$.
	
	(B) Now to prove the existence of $M'$ that leads to the power expression, consider the $\epsilon$ randomized response scheme.
	In particular, let $\cY = \cH$ and $\forall D \in \cH$, let
	\begin{align}
		\Pr\{M'(D) = D\} &= \frac{1}{1+ (k-1)e^{-\epsilon}} \\
		\Pr\{M'(D) = D'\} &= \frac{e^{-\epsilon}}{1+ (k-1)e^{-\epsilon}} \ \ D' \neq D.
	\end{align}
	The above mechanism is $\epsilon$-DP. Let $\psi$ be the identity map, then $\cP(\cM',\cH) \geq  \frac{1}{1 + (k-1)e^{-\epsilon}}$.
\end{proof}

\homoSwapCap*
\begin{proof}
	Let $D(S) = D_o + S$ where $S \in \cS(\cX)$ for $ |S| \leq t$ (including $S = \mset{\emptyset}$ to index $D_o$). 
	Define $e(S) = \Pr\{\psi(M(D(S))) \neq D(S)\}$.
	
	For a given mechanism $M$, and any $\psi$, we have
	\begin{equation}
		\Pr\{\psi(M(D(S))) = D(S) \} \leq  e^{|S|\epsilon} \Pr\{\psi(M(D_o)) = D(S)\},
	\end{equation}
	$\forall S \in \cS(\cX),  |S| \leq t$.
	Rearranging, we get
	\begin{align}
		e^{-|S|\epsilon} (1-e(S)) \leq \Pr\{\psi(M(D_o)) =  D(S)\}.
	\end{align}
	Adding over all choices of $S$ with $t \geq |S| \geq 1$, we get
	\begin{align}
		\sum_{S \in \cS(\cX): 0 < |S| \leq t}  e^{-|S|\epsilon} &\leq\sum_{S \in \cS(\cX): 0 \leq |S| \leq t}  e^{-|S|\epsilon} e(S).
	\end{align}

	Thus. similar to proof of \cref{claim:1}(A), we get
	\begin{align}
		\cP(\cM,\cH_t) &\leq \sum_{S \in \cS(\cX), |S| \leq t} (1-e(S))\pi_o(S) \\
		&= \frac{1}{ \sum_{S \in \cS(\cX), |S| \leq t} e^{-|S|\epsilon}}.
	\end{align}
	
	At this point, setting $t=1$, we get part (A2).
	For part (A1), note that
	\begin{align}
		\cP(\cM,\cH_t) &\leq \frac{1}{ \sum_{S \in \cS(\cX), |S| \leq t} e^{-|S|\epsilon}} \\
		&=  \frac{1}{  \sum_{S \in \cS(\cX), |S| \leq t} \prod_{x \in \cX} e^{-h_S(x)\epsilon}}.
	\end{align}
	Taking $\lim_{t \to \infty}$ and noting that 
	\begin{align}
		\sum_{S \in \cS(\cX), |S| \geq_0} \prod_{x \in \cX} e^{-h_S(x)\epsilon} &= \sum_{h_S(x) \geq 0 \forall x \in \cX} \prod_{x \in \cX} e^{-h_S(x)\epsilon} \\
		&=  \prod_{x \in \cX} \sum_{h_S(x) \geq 0}  e^{-h_S(x)\epsilon} \\
		&=  \prod_{x \in \cX} \frac{1}{1-e^{-\epsilon}}, 
			\end{align}
	we get the stated result. \\

	\textit{Converse direction:} (B1) follows similar to the converse proof of \cref{claim:1} using $\epsilon$ randomized response in add-remove model of neighbors. We focus on the converse of (A).
	
	Let $\cY = \cH_t$ and let $\cM'$ be such that 
	\begin{equation}
		\Pr\{\cM'(D) = D'\} \propto e^{-\epsilon d(D,D')},
	\end{equation}
	and let $\psi$ be the identity map.
	It is easy to verify that $M'$ is $\epsilon$-DP.
	Note that the mechanism $M'$ is $\epsilon$-DP on the whole of $\cS(\cX)$, not just $\cH_t$.
 	Thus, we have
	\begin{equation}
		\cP(\cM',\cH_t) \geq \min_{D \in \cH_t} \Pr_{D \sim \pi}\{\cM'(D) = D\}.
	\end{equation}
		Without loss of generality, denoting $\cX$ as $[k]$, we have $\forall D \in \cH_t$,
	\begin{align}
		\Pr\{M'(D) = D\} &= \frac{1}{\sum_{D' \in \cH_t} e^{-\epsilon d(D,D')}} \\
		&= \frac{1}{\sum_{D' \in \cH_t} \prod_{i=1}^k e^{-\epsilon |h_D (i) - h_{D'}(i)|}}.
	\end{align}
	Rewriting the above in a more convenient form : let $h_{D'}(i) = h_{D_o}(i) + v(i)$ then we can parameterize $D'$ by $v$ such that $v(i) \geq 0$, $\sum_{i=1}^k v(i) \leq t$. Let the set of all such $v$(s) be $\cV_t$.
	Then,
	\begin{align}
		\Pr\{M'(D) = D\} &= \frac{1}{\sum_{v \in \cV_t} \prod_{i=1}^k e^{-\epsilon |h_D (i) - h_{D_o}(i) - v(i)|}} \label{eq:B2}\\
		&\geq \frac{1}{\sum_{v \in \cV_{\infty}} \prod_{i=1}^k e^{-\epsilon |h_D (i) - h_{D_o}(i) - v(i)|}} \\
		&= \frac{1}{\prod_{i=1}^k \sum_{v(i) \geq 0} e^{-\epsilon |h_D (i) - h_{D_o}(i) - v(i)|}} \\		
&\geq \frac{1}{\prod_{i=1}^k \frac{1+e^{-\epsilon}}{1-e^{-\epsilon}}},  \label{eq:series}
	\end{align}
	where in \eqref{eq:series} we used $ \lim_{t \to \infty} \sum_{0 \leq v(i) \leq t} e^{-\epsilon |h_D (i) - h_{D_o}(i) - v(i)|} \leq  \lim_{t \to \infty}  \sum_{ -t \leq v(i) \leq t} e^{-\epsilon |v(i)|} = \frac{1+e^{-\epsilon}}{1-e^{-\epsilon}}$.
	Thus, we obtain $	\cP(\cM',\cH_t) \geq  \left( \frac{1+e^{-\epsilon}}{1-e^{-\epsilon}} \right)^k$. \\
	For (B2), use \eqref{eq:B2} with $t=1$ to get the stated result.
\end{proof}

\hetCap*
\begin{proof}
	Let $D(S) = D_o + S$ where $S \in \cS(\cW_o)$ for $ |S| \leq t$
	Note that we consider $S \in \cW_o$ in this proof, not $\cW$. 
	In essence, we shall come up with a prior on these hypotheses to upper bound the power.
	Define $e(S) = \Pr\{\psi(M(D(S))) \neq \Pi_{\cX}(D(S))\}$.
	Then, we AHDP guarantee, we have
	\begin{equation}
		e^{-d_{\alpha}(D(S),D_o)}(1-e(S)) \leq \Pr\{\psi(M(D_o)) = D(S)\} \ \forall S \in \cW_o, |S| \leq t.
	\end{equation}
	Summing over all $S \in \cW_o, 0 < |S| \leq t$, we get
	\begin{align}
		\sum_{S \in \cW_o, |S| \leq t, |S| \geq 1}e^{-d_{\alpha}(D_o, D_o + S)} \leq \sum_{S \in \cW_o, |S| \leq t}e^{-d_{\alpha}(D_o, D_o + S)} e(S).
	\end{align}
	Thus, upper bounding the power similar to \cref{claim:1}, we get
	\begin{align}
		\lim_{t \infty} \cP(M,\cH_t) &\leq 	\lim_{t \to \infty} \frac{1}{\sum_{S \in \cW_o, |S| \leq t}e^{-d_{\alpha}(D_o, D_o + S)}} \label{eq:cvx2} \\
		&= 	\lim_{t \to \infty} \frac{1}{\sum_{S \in \cW_o, |S| \leq t} \prod_{(x,\epsilon) \in \cW_o }e^{-\alpha(x,\epsilon)|h_{D_o +S}(x,\epsilon) - h_{D_o}(x,\epsilon)|}} \\
		&\leq 	\lim_{t \to \infty} \frac{1}{\sum_{S \in \cW_o, |S| \leq t} \prod_{(x,\epsilon) \in \cW_o }e^{-\epsilon |h_{D_o +S}(x,\epsilon) - h_{D_o}(x,\epsilon)|}} \label{eq:ade}\\
		&= \lim_{t \to \infty} \frac{1}{\sum_{S \in \cW_o, |S| \leq t} \prod_{(x,\epsilon) \in \cW_o }e^{-\epsilon |h_{S}(x,\epsilon)|}} \\
		&=  \frac{1}{ \prod_{(x,\epsilon) \in \cW_o } \sum_{h_S(x,\epsilon) \geq 0} e^{-\epsilon |h_{S}(x,\epsilon)|}} \\
		&=  \prod_{(x,\epsilon) \in \cW_o } (1- e^{-\epsilon}),
	\end{align}
	where in \eqref{eq:ade} we use the $\cW$-AHDP condition. 
    In \eqref{eq:cvx2}, setting $t=1$ leads to (A2). \\
    
	\textit{Converse direction:} \\ 
	  Let $\cY = \cH_{t,\cX} = \bigcup_{D \in \cH_t} \Pi_{\cX}(D)$.
	   Equivalently, $$\cH_{t,\cX} = \bigcup_{S \in \cS(\cW_o), |S| \leq t} \Pi_{\cX}(D_o + S).$$ 
	   Let $M'$ be such that
	  \begin{eqnarray}
	  	\Pr\{M'(D) = \Pi_{\cX}(D_o + S)\} \propto e^{-d'(\Pi_{\cX}(D),\Pi_{\cX}(D_o + S))},
	  \end{eqnarray}
	  where we define the distance $d'(F_1, F_2)$ for $F_1, F_2 \in \cS(\cX)$ as the distance using the most conservative possible privacy values for the given user data.
	  Concretely, let $\epsilon_l(x) = \min \{\epsilon| (x,\epsilon) \in \cW\}$, then $d'(F_1,F_2) = \sum_{x \in \cX} \epsilon_l(x) | h_{F_1 }(x) -  h_{F_2 }(x) |$
	  This distance measure can also be thought of as mapping $F_1, F_2$ to the space $\cS(\cW_o)$ and using $d_{\alpha}$ distance in that space with $\alpha(x,\epsilon) = \epsilon$.
	  
	  First, let's verify that this $M'$ is indeed $\cW$-AHDP. Observe that 
	  \begin{align}
	  	\frac{	\Pr\{M'(D_1) = \Pi_{\cX}(D_o + S)\}}{	\Pr\{M'(D_2) = \Pi_{\cX}(D_o + S)\}} &= \frac{e^{d'(\Pi_{\cX}(D_2),\Pi_{\cX}(D_o + S))}}{e^{-d'(\Pi_{\cX}(D_1),\Pi_{\cX}(D_o + S))}} \\
	  	&\leq e^{-\sum_{x \in \cX} r(x) |h_{\Pi_{\cX}(D_1)}(x) - h_{\Pi_{\cX}(D_2)}(x)|} \\
		&= e^{-\sum_{(x,\epsilon) \in \cX} r(x) |h_{D_1}(x,\epsilon) - h_{D_2}(x,\epsilon)|} \\ 
		&\leq e^{-\sum_{(x,\epsilon) \in \cX} \epsilon |h_{D_1}(x,\epsilon) - h_{D_2}(x,\epsilon)|}. 
	  \end{align}
	  Thus, $M'$ is $\cW$-AHDP. Note that this mechanism is $\cW$-AHDP over all inputs $D \in \cS(\cW)$, not just $\cH_t$.
	
	Let $\psi$ be the identity map. 
	We now analyze $\Pr\{M'(D) = \Pi_{\cX}(D)\}$. 
	Let $D = D_o + S$ for $S \in \cS(\cW), |S| \leq t$, we get
	\begin{align}
		&\Pr\{M(D_o + S) = \Pi_{\cX}(D_o + S)\} \\
		&= \frac{1}{\sum_{V \in \cS(\cW), |V| \leq t} e^{-d'(\Pi_{\cX}(D_o + S),\Pi_{\cX}(D_o + V))}} \\
	&= \frac{1}{\sum_{V \in \cS(\cW), |V| \leq t} \prod_{x \in \cX} e^{-\epsilon_l(x)|h_{D_o + S}(x) - h_{D_o + V}(x)|}} \\				
	&= \frac{1}{\sum_{V \in \cS(\cW), |V| \leq t} \prod_{x \in \cX} e^{-\epsilon_l(x)|h_{S}(x) - h_{V}(x)|}} \label{eq:cvx3} \\
	&\geq \frac{1}{\sum_{V \in \cS(\cW): h_V(x) \geq 0 \forall x \in \cX} \prod_{x \in \cX} e^{-\epsilon_l(x)|h_{S}(x) - h_{V}(x)|}} \\
	&= \frac{1}{ \prod_{x \in \cX} \sum_{h_V(x) \geq 0}  e^{-\epsilon_l(x)|h_{S}(x) - h_{V}(x)|}} \\
	&\geq  \frac{1}{ \prod_{x \in \cX} \sum_{h_V(x) = -\infty}^{\infty}  e^{-\epsilon_l(x)|h_{S}(x) - h_{V}(x)|}} \\
	&= \prod_{x \in \cX} \left( \frac{1 - e^{-\epsilon_l(x)}}{1 + e^{-\epsilon_l(x)}}\right) \\
	&= \prod_{(x,\epsilon) \in \cW_o} \left( \frac{1 - e^{-\epsilon}}{1 + e^{-\epsilon}}\right).
	\end{align}
	Thus, we get the stated lower bound on power.
    In \eqref{eq:cvx3}, setting $t=1$ leads to (B2).
\end{proof}

\hetCapSec*
\begin{proof}
	(A) By AHDP guarantee, we have 
	\begin{equation}
		1 - \Pr\{ \psi(M(D_o + \mset{(x,\epsilon)})) = D_o  \} \leq e^{-\epsilon}  \Pr\{ \psi(M(D_o)) =   D_o + \mset{(x,\epsilon)} \}.
	\end{equation}	
	We thus get the upper bound on power $\frac{1}{1+e^{-\epsilon}}$. \\	
	(B) Let $\cY = \cH(x,\epsilon)$, $\alpha(x,\epsilon) = \epsilon$, $\Pr\{M'(D) = D'\} \propto e^{-d_{\alpha}(D,D')}$, and $\psi$ is the identity function.
	Note that $M'$ is $\cW$-AHDP.
	Thus, 
	$\Pr\{M'(D_o) = D_o\} = \frac{1}{1+e^{-\alpha(x,\epsilon)}} = \frac{1}{1 + e^{-\epsilon}} $.
	Similarly, $\Pr\{M'(D_o + \mset{(x,\epsilon)} ) = D_o  + \mset{(x,\epsilon)} \} = \frac{1}{1 + e^{-\epsilon}} $.
	Therefore, $\cP(M',\cH(x,\epsilon)) \geq \frac{1}{1+e^{-\epsilon}}$. 
\end{proof}
\section{Sample Mechanism: Proof of Privacy} \label{sec:SM}

We provide the proof of privacy for Sample Mechanism. 
The proof is quite similar to \citet{Jorg15}.
Represent the sub-sampling mechanism as $C: \cS(\cW) \to \cS(\cW)$ and let $\phi_t: \cS(\cW) \to \cY$ be the homogeneous $t$-DP mechanism.

Without loss of generality, we may prove the AHDP guarantee for neighboring datasets $D,D'$.
In particular, let $D' = D + \{\{(x,\epsilon)\}\}$.
We shall show $|  \log \frac{\Pr\{\phi_t(C(D')) \in S\}}{\Pr\{\phi_t(C(D)) \in S\}} | \leq \alpha(x,\epsilon) \wedge t$ for any $S \subseteq \cY$.

\begin{align}
	\Pr\{\phi_t(C(D')) \in S\} &= \Pr\{\phi_t(C(D)) \in S\} \Pr\{(x,\epsilon) \text{ not sampled}\}\\ 
	&+  \Pr\{\phi_t(C(D) + \{\{(x,\epsilon)\}\}) \in S\} \Pr\{(x,\epsilon) \text{ sampled}\} \\
	&=  \frac{e^t - e^{\alpha(x,\epsilon) \wedge t}}{e^t - 1} \Pr\{\phi_t(C(D)) \in S\} \\ 
	&+   \frac{e^{\alpha(x,\epsilon) \wedge t}-1}{e^t - 1} \Pr\{\phi_t(C(D) + \{\{(x,\epsilon)\}\}) \in S\} \\
&\leq  \frac{e^t - e^{\alpha(x,\epsilon) \wedge t}}{e^t - 1} \Pr\{\phi_t(C(D)) \in S\} \\ 
&+   \frac{e^{\alpha(x,\epsilon) \wedge t}-1}{e^t - 1} e^t\Pr\{\phi_t(C(D)) \in S\} \\
&= e^{\alpha(x,\epsilon) \wedge t} \Pr\{\phi_t(C(D)) \in S\}.
\end{align}

To show $ \frac{\Pr\{\phi_t(C(D')) \in S\}}{\Pr\{\phi_t(C(D)) \in S\}} \geq e^{- (\alpha(x,\epsilon) \wedge t)}$, we note that similarly,
\begin{align}
	\Pr\{\phi_t(C(D')) \in S\} &= \Pr\{\phi_t(C(D)) \in S\} \Pr\{(x,\epsilon) \text{ not sampled}\}\\ 
	&+  \Pr\{\phi_t(C(D) + \{\{(x,\epsilon)\}\}) \in S\} \Pr\{(x,\epsilon) \text{ sampled}\} \\
	&=  \frac{e^t - e^{\alpha(x,\epsilon) \wedge t}}{e^t - 1} \Pr\{\phi_t(C(D)) \in S\} \\ 
	&+   \frac{e^{\alpha(x,\epsilon) \wedge t}-1}{e^t - 1} \Pr\{\phi_t(C(D) + \{\{(x,\epsilon)\}\}) \in S\} \\
	&\geq  \frac{e^t - e^{\alpha(x,\epsilon) \wedge t}}{e^t - 1} \Pr\{\phi_t(C(D)) \in S\} \\ 
	&+   \frac{e^{\alpha(x,\epsilon) \wedge t}-1}{e^t - 1} e^{-t}\Pr\{\phi_t(C(D)) \in S\} \\
	&=\left( \frac{e^t - e^{-t}}{e^t - 1} - \frac{e^{\alpha(x,\epsilon) \wedge t}}{e^t} \right) \Pr\{\phi_t(C(D)) \in S\}.
\end{align}
Now, we claim that 
\begin{equation}
	\frac{e^t - e^{-t}}{e^t - 1} - \frac{e^{y \wedge t}}{e^t} \geq e^{-(y \wedge t)} \ \ \forall y,t \geq 0, 
\end{equation}
which would prove the desired privacy guarantee.
For $y \geq t$, the claim is trivially true. 
For $0 \leq y < t$, we have 
\begin{align}
	\frac{e^t - e^{-t}}{e^t - 1} - \frac{e^{y \wedge t}}{e^t} &= \frac{e^t - e^{-t}}{e^t - 1} - \frac{e^{y}}{e^t} \\
	&=  1 + \frac{1}{e^t} - \frac{e^{y}}{e^t} \\
	&=  e^{y}(e^{-y} - e^{-t}) + e^{-t} \\
	&\geq 	e^{-y} - e^{-t} + e^{-t} \\
	&= e^{-(y \wedge t)},
\end{align}
where we used $e^y \geq 0$.

%% file: B-App.tex
\section{Linear Regression - Proof of Privacy} \label{sec:LRP}

We shall show that the $(i,j)$-th element of $A$ is $\alpha/(d^2 + d)$-AHDP.

Suppose $D' = D + \mset{(z,\epsilon)}$ and let datapoint be the $n+1$-th element of $D'$.
Let output matrix $A$ be $A_D$ and $A_{D'}$ for input datasets $D$ and $D'$ respectively.
We have
\begin{align*}
    \frac{\Pr\{A_D(i,j) =s \} }{\Pr\{ A_{D'}(i,j) =s \} } &= \exp\{ -\frac{1}{d^2 + d}|s - \sum_{k=1}^n x_{ki}x_{kj}w_{k}| \} \\
    &\times \exp\{ \frac{1}{d^2 + d}|s - \sum_{k=1}^{n+1} x_{ki}x_{kj}w_{k}| \} \\
    &\leq \exp\{ \frac{1}{d^2 + d} | x_{n+1,i}x_{n+1,j}\alpha(z,\epsilon)|\} \\
    &\leq \exp\{ \frac{\alpha(z,\epsilon)}{d^2 + d}\}.
\end{align*}

Similarly, $i$-th element of $b$ is $\alpha/(d^2 + d)$-AHDP, and basic composition yields the overall AHDP guarantee.

\section{Experiment Details} \label{sec:exp-det}

We used OpenAI's GPT-4o \cite{hurst2024gpt} for the experiments.
For the first experiment, we queries the LLM with the following prompt.
\begin{quote}
    I am collecting user data which is your weight. 
    I will use this for my algorithm, but you are allowed to tell me the privacy level you want in terms of epsilon in epsilon-DP. 
    That is, if you ask for x level, I will provide your data (weight) with x-DP privacy protection in the framework of differential privacy.  
    Do this exercise 10 times: sample your weight in kg (2 decimal points) from a reasonable distribution. Based on your weight, think how much privacy you want - do you want more privacy or less? 
    Only return the tuples of (weight, epsilon) for the 10 samples, new line for each sample.
    Your privacy demand can range from 0 to 3.
\end{quote}
We then iterated the same experiment several times to generate the dataset. The presence of randomness in the LLM outputs along with setting the temperature higher ensured that we did not obtain duplicated data.
It was not feasible to get the LLM to generate hundreds of datapoints in one query as it would switch to making a python script to generate the dataset.

For the second experiment, we queried the LLM with the following prompt several times
\begin{quote}
    Imagine you are a human user and I am a service provider.
    I am collecting user data which is your education level. 
    The possible values are: 1 = no education, 2 =  primary school, 3 =  secondary school, 4 = bachelor, 5 =  master, 6 = PhD.
    I offer you 5 privacy levels in the framework of Differential Privacy.
    Recall that low epsilon means high privacy and high epsilon means low privacy. 
    The epsilon values are: 0.01, 0.1, 0.5, 1, and 5.
    Do this exercise 10 times: sample your education level from a reasonable distribution (location: random country in the world). Based on your education level, think how much privacy you want - do you want more privacy or less? 
    Only return the tuples of (education level, epsilon) for the 10 samples, new line for each sample.
\end{quote}

\textbf{Frequency estimation details:} Suppose $\cX = [k]$, then, for a dataset $D$, compute $S(D) = \sum_{(x,\epsilon) \in T(D)}\alpha_2(x,\epsilon) + \cL(1)$.
This is $\alpha_2$-AHDP estimate of the size of the dataset.
For each $i \in [k]$, independently compute $N_i(D) = \sum_{(x,\epsilon) \in T(D)}h_D(x,\epsilon)\alpha_1(x,\epsilon) \bbI\{x = i\} + \cL(1)$.
The mechanism output is $N_i(D)/S(D)$ for the $i$-th bin.

To show AHDP propert, consider $D' = D + \mset{(a,\epsilon)}$.
For any $t_1,\ldots,t_k$
\begin{align*}
    \frac{\Pr\{N(D)_i  = t_i \forall i \in [k] \}}{\Pr\{N(D')_i  = t_i \forall i \in [k] \}} &= \prod_i e^{|t_i - \sum_{(x,\epsilon) \in T(D')}h_{D'}(x,\epsilon)\alpha_1(x,\epsilon) \bbI\{x = i\}|} \\
    & \times \prod_i e^{- |t_i - \sum_{(x,\epsilon) \in T(D)}h_D(x,\epsilon)\alpha_1(x,\epsilon) \bbI\{x = i\}| } \\
    &\leq \prod_i e^{\sum_{(x,\epsilon) \in \cW}\alpha_1(x,\epsilon) \bbI\{x = i\}|h_D(x,\epsilon) - h_{D'}(x,\epsilon)|} \\
    &\leq e^{\alpha_1(a,\epsilon)}.
\end{align*}
Thus, releasing $N_1(D),\ldots,N_k(D)$ independently is $\alpha_1$-AHDP.
By composition, the proposed mechanism is $\alpha_1 + \alpha_2$-AHDP.